\documentclass[11pt]{amsart}
\usepackage{amsmath, amssymb, amsbsy, amsfonts, amsthm, latexsym, amsopn, amstext, amsxtra, euscript, amscd, color, mathrsfs}
\usepackage[normalem]{ulem}
\usepackage{soul}

\PassOptionsToPackage{hyphens}{url}\usepackage{hyperref}
\makeatletter
\@namedef{subjclassname@2020}{%
  \textup{2020} Mathematics Subject Classification}
\makeatother
 
 \usepackage[capbesideposition=outside,capbesidesep=quad]{floatrow}

\restylefloat{table}
\restylefloat{table}
            
\usepackage{multirow,caption}
            
\usepackage{amscd}
\usepackage{color,enumerate}
\newcommand{\RNum}[1]{\lowercase\expandafter{\romannumeral #1\relax}}

\usepackage[colorinlistoftodos,prependcaption,textsize=tiny]{todonotes}

\newtheorem{thm}{Theorem}[section]
\newtheorem{lem}[thm]{Lemma}

\newtheorem{rem}[thm]{Remark}

\newtheorem{thm-con}[thm]{Theorem-Conjecture}
\numberwithin{equation}{section}

\theoremstyle{definition}
\newtheorem{defn}[thm]{Definition}

\def\cG{{\mathcal G}}

\newcommand{\F}{\mathbb F}

\def\Tr{{\rm Tr}}

\def\Trmn{{\rm Tr}_m^{3m}}
\def\Trnn{{\rm Tr}_m^{2m}}
\def\Trmm{{\rm Tr}_m^{n}}
\def\Trnl{{\rm Tr}_{\ell}^{n}}

\begin{document}
\title[Differential uniformity of some permutation polynomials]{Differential uniformity properties of some classes of permutation polynomials}

 \author[K. Garg]{Kirpa Garg}
 \address{Department of Mathematics, Indian Institute of Technology Jammu, Jammu 181221, India}
 \email{kirpa.garg@gmail.com}
 
 \author[S. U. Hasan]{Sartaj Ul Hasan}
 \address{Department of Mathematics, Indian Institute of Technology Jammu, Jammu 181221, India}
 \email{sartaj.hasan@iitjammu.ac.in}
 
  \author[P.~St\u anic\u a]{Pantelimon~St\u anic\u a}
  \address{Applied Mathematics Department, Naval Postgraduate School, Monterey, CA 93943, USA}
 \email{pstanica@nps.edu}

%
\begin{abstract}
The notion of $c$-differential uniformity has recently received a lot of attention since its proposal~\cite{Ellingsen}, and recently a characterization of perfect $c$-nonlinear functions in terms of difference sets in some quasigroups was obtained in~\cite{AMS22}. Independent of their applications as a measure for certain statistical biases, the construction of functions, especially permutations, with low $c$-differential uniformity is an interesting mathematical problem in this area, and recent work has focused heavily in this direction. We provide a few classes of permutation polynomials with low $c$-differential uniformity. The used technique  involves handling various Weil sums, as well as analyzing some equations in finite fields, and we believe these can be of independent interest.
\end{abstract}

\keywords{Finite fields, permutation polynomials, $c$-differential uniformity}

\subjclass[2020]{12E20, 11T06, 94A60}

\maketitle
\section{Introduction} Let $p$ be a prime number and $n$ be a positive integer. We denote  by $\F_{q}$ the finite field with $q$ elements, by $\F_{q}^*$ the multiplicative group of non-zero elements of $\F_{q}$ and by $\F_{q}[X]$ the ring of polynomials in one variable $X$ with coefficients in $\F_{q}$, where $q=p^n$.  Let $F$ be a function from $\F_{q}$ to itself. 
Lagrange's interpolation formula allows us to uniquely represent $F$ as a polynomial in $\F_q[X]$ of degree at most $q-1$. A polynomial $F \in \F_{q}[X]$ is a permutation polynomial of $\F_{q}$ if the mapping $X \mapsto F(X)$ is a permutation of $\F_{q}$.  It is worth emphasising that due to their numerous applications in coding theory~\cite{Chapuy_C_07, Ding_C_13}, combinatorial design theory~\cite{Ding_Co_06}, cryptography~\cite{Lidl_Cr_84, Schwenk_Cr_98}, and other branches of mathematics and engineering, permutation polynomials over finite fields are highly significant objects. These functions, for instance, are frequently used in cryptography to construct substitution boxes (S-boxes), which are a key component of contemporary block ciphers.

There are many known attacks on block ciphers. One of the most powerful attacks on block ciphers is differential cryptanalysis, which was first developed by Biham and Shamir~\cite{{Biham91}}. The concept of differential uniformity was first introduced by Nyberg~\cite{Nyberg} to measure a function's resistance to the differential attack, and it is defined as follows. For any function $F : \F_{q}\to \F_{q}$ and for any $a \in \F_{q}$, the derivative of $F$ in the direction $a$ is defined as $D_F(X,a) := F(X+a)-F(X)$ for all $X\in \F_{q}.$ The Difference Distribution Table (DDT) entry of $F$ at a point $(a,b) \in \F_{q} \times \F_{q}$, denoted by $\Delta_F(a,b)$, is the number of solutions $X\in \F_{q}$ of the equation $D_F(X,a) = b$. The differential uniformity of $F$, denoted by $\Delta_F$, is given by $\Delta_{F} := \max\{\Delta_{F}(a, b) : a\in \F_{q}^*, b \in \F_{q} \}.$ When $\Delta_{F} = 1$, $F$ is called perfect nonlinear (PN) function. When $\Delta_{F} = 2$, $F$ is called almost perfect nonlinear (APN) function. It should be noted that there are no PN functions over finite fields with even characteristic.

The multiplicative differentials of the form $(F(cX), F(X))$ were introduced by Borisov et al.~\cite{Borisov} who exploited this new class of differential to attack certain existing ciphers. Ellingsen et al.~\cite{Ellingsen} extended on the idea of differential uniformity and developed a new (output) multiplicative differential as a result of the multiplicative differential. For any function $F : \F_{q}\to \F_{q}$ and for any $a , c \in \F_{q}$, the (multiplicative) $c$-derivative of $F$ with respect to $a$ is defined as $_c\Delta_F (X, a) := F(X + a)- cF(X)$ for all $X \in \F_{q}$. For any $a,b \in \F_q$, the $c$-Difference Distribution Table ($c$-DDT) entry $_c\Delta_F(a,b)$ at point $(a,b)$ is the number of solutions { $X \in \F_q$} of the equation $_cD_F(X,a) = b$. The $c$-differential uniformity of $F$, denoted by $_c\Delta_F$, is given by $_c\Delta_{F} := \max\{_c\Delta_{F}(a, b) : a, b \in \F_{q} \hspace{0.2cm}\text{and}\hspace{0.2cm} a \neq 0 \hspace{0.2cm}\text{if} \hspace{0.2cm}c=1\}.$ It is clear that when $c = 1$, differential uniformity and $c$-differential uniformity are same. We refer to $F$ as a perfect $c$-nonlinear (P$c$N) function and an almost perfect $c$-nonlinear (AP$c$N) function, respectively, for $_c\Delta_F= 1$ and $_c\Delta_F=2$. Note that for monomial functions, $X \mapsto X^d$, the output differential  $(c_1F(X), F(X))$ is the same as the input differential  $(F(c_2X), F(X))$, where $c_1=c_2^d$. 

The authors of~\cite{BKM22} discuss the potential of an extension of the differential attack based upon $c$-differentials, and show that a large class of potential $S$-boxes have large $c$-differential uniformity for all but few choices of $c$. In view of this, these statistical biases are somewhat inevitable. However, as they point out  ``{\em the $c$-differential uniformities still measure biases in the distribution of differences and it might still theoretically be possible to construct an attack different than the one considered (t)here to abuse this bias}''~\cite{BKM22}.

As a first application of the concept of $c$-differential uniformity, we  point out that in a recent manuscript~\cite{AMS22}, the graph of a P$c$N function was shown to correspond to a difference set in a quasigroup. 
Difference sets give rise to symmetric designs, which are known to construct optimal self complementary codes. Various types of designs can be also used in secret sharing and visual cryptography (see also~\cite{XZG1,XZG2}, where it is shown that difference sets can be used to construct a complex vector codebook that achieves the Welch bound on maximum crosscorrelation amplitude).

Finding functions, particularly permutations, with low $c$-differential uniformity has received a lot of interest since the concept of $c$-differential uniformity was established. In~\cite{HPRS20, JKK2, Liu, MRS, wang, lwz, Hu}, numerous functions with low $c$-differential uniformity were investigated. Only a few P$c$N and AP$c$N functions are known over a finite field with even characteristic; see, for example, \cite{HPS1, JKK,  ZJT}. Recently, Li et al.~\cite{CPC} extended Dillon’s switching method to $c$-differentials and applied it to find necessary and sufficient conditions for such a constructed function to be P$c$N or AP$c$N, as well as to generalize it to any $c$-differential uniformity. Further, using this technique, the authors give some classes of P$c$N and AP$c$N functions as well. In this paper, we study the $c$-differential uniformity of some classes of permutation polynomials introduced in~\cite{LWC}. The paper is organised as follows. In Section~\ref{S2}, we recall some relevant results that are required in the subsequent sections.  The $c$-differential uniformity of two classes of  permutation polynomials over finite fields of even characteristic has been considered in Section~\ref{S3}. Further, Section~\ref{S4} deals with the $c$-differential uniformity of two classes of  permutation polynomials over finite fields of odd characteristic. Finally, we conclude the paper in Section~\ref{S5}.

\section{Preliminaries} \label{S2}
In this section, we first review a definition and provide some lemmas to be used later. In what follows, we shall use $\Trmm$ to denote the (relative) trace function from $\F_{p^n} \rightarrow \F_{p^m}$, i.e., $\Trmm(X) = \sum_{i=0}^{\frac{n-m}{m}} X^{p^{mi}}$, where $m$ and $n$ are positive integers and $m|n$.  When $m=1$, we use $\Tr$ to denote the absolute trace.

\begin{defn}\cite{TH}
 For a function $F : \F_{p^n} \rightarrow \F_{p}$, the Walsh transform of $F$ at $v \in  \F_{p^n}$, is defined as
$$ \mathcal{W}_F(v) = \sum_{X \in \F_{p^n}} \omega^{F(X) - \Tr(v X)},$$
where $\omega = e^{\frac{2 \pi i}{p}}$ is the complex primitive $p$th root of unity. 
\end{defn}

\begin{lem} \label{L00}\cite[Proposition 4]{LWC}
 For a positive integer $m$ and a fixed $\delta$ in $\F_{2^{3m}}$, the polynomial $$ F(X) = (X^{2^m}+X+\delta)^s+X $$ is a permutation of $\F_{2^{3m}}$ if one of the following conditions holds:
 \begin{enumerate}
  \item[$(1)$] $s= 2^{2m}+1 ;$
  \item[$(2)$] $s= 2^{im-1}+ 2^{m-1}, \gcd((i-1)m-1,3m)=1, i\in \{2, 3\}.$
 \end{enumerate}
\end{lem}

\begin{lem} \label{L02}\cite[Proposition 3]{LWC}
For an even positive integer $m$ and a fixed $\delta \in \F_{3^{2m}}$, the polynomial $$F(X) = (X^{3^m}-X+\delta)^{3^{2m-1} + 2 \cdot 3^{m-1}}+X $$
is a permutation of $\F_{3^{2m}}$.
\end{lem}

\begin{lem} \label{L03}\cite[Proposition 2]{LWC}
 For an odd prime $p$ and a positive integer $m$, let $\delta$ be an element of $\F_{p^{2m}}$ such that $\Trnn(\delta)=0$ or $\dfrac{\Trnn(\delta)+1}{\Trnn(\delta)}$ is a $(p-1)$-th power in $\F_{p^m}$, then the polynomial
$F(X) = (X^{p^m}-X+\delta)^{p^{m+1} + 1}+X$ is a permutation of $\F_{p^{2m}}$.
\end{lem}

\begin{lem}\label{L01}\cite[Theorem 12]{CPC}
Let $m$ be a positive integer and $n=2m$ such that $n\geq 3$. Furthermore, let $L$ be a linearized permutation polynomial such that $L(-1) = -1$ and let $G_{k_{1} \leq k_{2} \leq \cdots \leq k_{s}}(X) := L(X)+ \displaystyle{\prod_{i=1}^{s} (\alpha_{i} \Trmm(X^{2^{k_{i}} +1} + \delta_{i}))^{g_i} } \in \F_{2^n}[X]$, where $g_{i} \in \mathbb{N}$, $\delta_{i} \in \F_{2^n}$, $\alpha_{i} \in \F_{2^m}^{*}$ and $1 \leq k_{i} \leq n-1$. Then $G_{k_{1} \leq k_{2} \leq \cdots \leq k_{s}}$ is either P$c$N or AP$c$N with respect to all $c \neq 1$, and P$c$N for $c=0$.
\end{lem}
The following lemma can be gleaned from the proof of ~\cite[Proposition 2]{TH}.
 
\begin{lem} \label{walsh} Let $m$ be a positive integer and $n=2m$. Also, let $a_i \in \F_{p^n} (i=0, \cdots, m$) for an odd prime $p$. Then the absolute square of Walsh transform coefficient of the function $f : X \mapsto \Tr\left( \sum_{i=0}^{m} a_i X^{p^i+1}\right)$ at $-v \in \F_{p^n}$ is given by
 \begin{equation*} \lvert \mathcal{W}_f(-v) \rvert^2 =
  \begin{cases}
   p^{n+\ell} &~\mbox{if}~f(X)+\Tr(vX)\equiv0~\text{on Ker}~(L)  \\
    0 &~\mbox{otherwise},
  \end{cases}
 \end{equation*}
 where $\ell$ is dimension of kernel of the linearized polynomial $L(X) = \sum_{i=0}^{m} (a_i X^{p^i} +{(a_i X)}^{p^{n-i}}).$ 
\end{lem}

We also recall that in~\cite{StanicaPG2021}, the authors computed the $c$-DDT entries by means of the Weil sums approach. We will quickly go over the general technique for expressing the number of solutions to a given equation over finite fields in terms of Weil sums for the reader's convenience. Let $\chi_1: \F_q \rightarrow \mathbb{C}$ be the canonical additive character of the additive group of $\F_q$ defined as follows
\[
 \chi_1(X):= \exp \left(\frac{2\pi i \Tr(X)}{p} \right).
\]
It is easy to observe (see, for instance~\cite{PSweil}) that the number of solutions $(X_1, X_2, \ldots, X_n) \in \F_q^n$ of the equation $$F(X_1, X_2, \ldots, X_n)=b,$$ denoted by $N(b)$, is given by 
\begin{equation}
\label{ddtw}
 \begin{split}
 N(b)= \frac{1}{q} \sum_{X_1,X_2, \ldots, X_n \in \F_q} \sum_{\beta \in \F_q} \chi_1(\beta(F(X_1, X_2, \ldots, X_n)-b)).
 \end{split}
\end{equation}
In the sections that follow, the expression from Equation~\eqref{ddtw} will be used to determine the c-differential uniformity of a few permutations over finite fields.

\section{Permutations over $\F_{2^n}$ with low $c$-differential uniformity}\label{S3}
In this section, we first deal with the computation of the $c$-differential uniformity of $F(X)=(X^{2^m}+X+\delta)^{2^{2m}+1}+X $ over $\F_{2^n}$, where $n=3m$ and $\delta \in \F_{2^n}$. From Lemma~\ref{L00}, we know that $F$ is a permutation polynomial over $\F_{2^n}$. Here we find conditions on $c$ and $\delta$ for which $F$ turns out to be either a P$c$N or an AP$c$N function. Notice that in our case, when $\delta=1$, the function $F(X)=(X^{2^m}+X+\delta)^{2^{2m}+1}+X $ can be rewritten as $F(X) = L(X) + \Trmn(X^{2^m +1}+1)$, where $L(X) = X^2 + \Trmn(X)$. As $L$ is a non-permutation over $\F_{2^n}$ with $L(1)=0$, our case is different from the function discussed in Lemma~\ref{L01}.  Moreover, the method developed in~\cite{CPC}  cannot be used to treat our class of functions.

Our first theorem will be based on two lemmas that we will now prove. From~\cite[Theorem 4]{LMW}, the following lemma can be extracted directly, but we add its proof here for completeness.

\begin{lem} \label{lemma1S1}
Let $u \in \F_{2^m}^{*}$ and $G(X):=uX^{2^m+1}$ be a function on $\F_{2^n}$, where $n=3m$. Then $
\mathcal{W}_G(v)=0~\text{if}~ \Tr (v)=0.$
\end{lem}
\begin{proof}
Let $\zeta$ be an  element of $\F_{2^n}$, to be determined later. Surely, $X\mapsto X+\zeta$ is a bijection of $\F_{2^n}$.
Then,
\allowdisplaybreaks
\begin{align*}
 \mathcal{W}_G(v) & = \sum_{X \in \F_{2^n}} \chi (uX^{2^m+1}+vX) \\
 & = \sum_{X \in \F_{2^n}} \chi (u(X+\zeta)^{2^m+1}+v(X+\zeta)) \\
 & = \sum_{X \in \F_{2^n}} \chi (uX^{2^m+1}+u\zeta^{2^m} X+ u\zeta X^{2^m}+u\zeta^{2^m+1}+vX+v\zeta)\\
& =\chi(u \zeta^{2^m+1}+v\zeta) \sum_{X \in \F_{2^n}} \chi (uX^{2^m+1}+u\zeta^{2^m} X+ u\zeta X^{2^m}+vX)\\
& = \chi(u \zeta^{2^m+1}+v\zeta) \sum_{X \in \F_{2^n}} \chi (uX^{2^m+1}+X(L(\zeta)+ v)), 
\end{align*}
 where $L(\zeta)= u \zeta^{2^m} + (u \zeta)^{2^{-m}}$. Notice that $L$ is $2-$to$-1$ map and the image of $L$ is the set of elements in $\F_{2^n}$ of trace 0, since $u \in \F_{2^m}$. So, if $\Tr(v) = 0$, then we can choose $\zeta$ such that $L(\zeta) = v$, and the sum becomes $\sum_{X \in \F_{2^n}}\chi(uX^{2^m+1})$ which is $0$, as $u \neq 0$ and $\gcd(2^m+1,2^{3m}-1)=1$. Hence,   $\mathcal{W}_G(v)=0$ if $\Tr(v)=0$.
\end{proof}
The maximum number of solutions to a certain equation that arises in the proof of the next theorem are given by the following lemma, which may be of some independent interest.
\begin{lem} 
\label{lemma2S1}
 Let $m$ be a positive integer, $a  \in \F_{2^n}$, where $n=3m$. Furthermore, let $\delta \in \F_{2^n}$ with $\Trmn (\delta) \neq 1$ and $c \in \F_{2^n} \setminus \F_{2^m}$. Then the following equation
\begin{equation}
\label{eq:case3}
((1+c)X)^{2^{-1}}+(1+c)(1+\Trmn(\delta))X+ (a^{2^{2m}} +a^{2^m})\Trmn(X)=0
\end{equation}
has at most four solutions in $\F_{2^n}$ under the restriction that $\Trmn ((1+c)X)=0$.
\end{lem}
\begin{proof}
It is easy to see that for those $a \in \F_{2^n}$ satisfying $a^{2^{2m}} +a^{2^m}=0$, the above equation has two solutions.

Now, we consider $a^{2^{2m}} +a^{2^m} \neq 0$ and obtain the following equation 
\begin{equation}
\label{eq:case3}
\frac{(1+c)^{2^{-1}}}{a^{2^{2m}}+a^{2^m}}X^{2^{-1}} +\frac{(1+c)(1+\Trmn(\delta))}{a^{2^{2m}} +a^{2^m}}X+ \Trmn(X)=0.
\end{equation}

We shall proceed in a ``multivariate approach''  way. For simplicity, we let $\delta'=1+\Trmn(\delta)\in\F_{2^m}$, $A =a^{2^m}+a^{2^{2m}}\not\in\F_{2^m}$, $\epsilon_1=\frac{(1+c)\delta'}{A}, \epsilon_2=\frac{(1+c)^{2^{-1}}}{A}$. We now replace $\Trmn(X)=X+X^{2^m}+X^{2^{2m}}$ in  Equation~\eqref{eq:case3} 
and obtain
\[
X^{2^{2m}}+X^{2^m}+\left(\epsilon_1+1 \right)X+\epsilon_2X^{2^{-1}}=0.
\]
Raising this  to the $2^m$-power, we get
\begin{equation}
\label{eq:case3_0}
X^{2^{2m}}+\left(\epsilon_1^{2^m} +1 \right)X^{2^m}+\epsilon_2^{2^m}X^{2^{m-1}}+X=0.
\end{equation}
We raise again to the $2^m$-power,   obtaining
\begin{equation}
\label{eq:case3_1}
 \left(\epsilon_1^{2^{2m}} +1 \right)X^{2^{2m}}+\epsilon_2^{2^{2m}}X^{2^{2m-1}}+X^{2^m}+X=0.
\end{equation}
Now, adding Equations~\eqref{eq:case3_0} and~\eqref{eq:case3_1}, we get
\[
\epsilon_1^{2^{2m}}  X^{2^{2m}}+\epsilon_2^{2^{2m}}X^{2^{2m-1}}+\epsilon_1^{2^m}  X^{2^m}+\epsilon_2^{2^m}X^{2^{m-1}}=0,
\]
which is equivalent to
\[
(\epsilon_1^2X^2+\epsilon_2^2X)^{2^{2m-1}}+(\epsilon_1^2X^2+\epsilon_2^2X)^{2^{m-1}}=0.
\]
Taking the $2^{m-1}$ root in the above equation, it follows that 
\[
\epsilon_1^2X^2+\epsilon_2^2X =\alpha \in\F_{2^m},
\]
which can be written as
\[
((1+c)X \delta')^2+(1+c)X+\alpha A^2=0.
\]
Using the substitution $u=(1+c)X\delta'$, we obtain the equation $(\delta' \neq 0)$,
\begin{equation}
\label{eq:equ}
u^2+\frac{u}{\delta'}+\alpha A^2 =0.
\end{equation}
This equation has two distinct solutions $u,u+\frac{1}{\delta'}$ (for a fixed $\alpha$) if and only if $\Tr\left(\alpha \delta'^2 A^2  \right)=0$. Observe that $\Tr\left(\alpha  \delta'^2 A^2  \right)=\Tr_1^m\left(\Trmn \left(\alpha  \delta'^2 A^2  \right)\right)=\Tr_1^m\left(\alpha  \delta'^2\Trmn \left( A^2 \right)\right)=0$ (since $\Trmn \left( A^2 \right)=0$), so the condition for the existence of solutions of \eqref{eq:equ} is automatically satisfied.
The needed condition on the solutions, that is, $\Trmn((1+c)X)=0$ is surely equivalent to $\Trmn(u)=0$ (since $\delta'\in\F_{2^m}$). It follows that at most one solution $u,u+\dfrac{1}{\delta'}$ can satisfy $\Trmn(u)=0$, or $\Trmn(u+\frac{1}{\delta'})=0$, since $\Trmn(1)=1$.

We next observe that $\alpha=\left(\dfrac{1}{\delta'}\Trmn\left(\frac{u}{1+c} \right)\right)^2$ and so, we need to show that the following linearized polynomial
\begin{equation}
\label{eq:Lu}
L(u)=u^2+\frac{u}{\delta'}+\frac{A^2}{\delta'^2}\Trmn\left( \frac{1}{1+c} u\right)^2
\end{equation}
has at most four solution in $\F_{2^n}$ for all $A$.  
 
 Since $\F_{2^{3m}}$ is an extension of degree 3 over $\F_{2^m}$, it follows that any element in $\F_{2^{3m}}$ that does not belong to $\F_{2^{m}}$ generates $\F_{2^{3m}}$ over $\F_{2^{m}}$. In particular, $\{1,A,A^2\}$, where $A=a^{2^m}+a^{2^{2m}}$, forms a basis of $\F_{2^{3m}}$ over $\F_{2^{m}}$. Surely, there is a polynomial  (over $\F_{2^{m}}$)  of degree 3 with root $A$, say $A^3+\gamma_2 A^2+\gamma_1 A+\gamma_0=0$, $\gamma_0\neq 0$. Recall that $\Trmn(A)=0$.
 We are looking for (the number of) solutions $u$ of $L(u)=0$, with $\Trmn(u)=0$. We write 
\begin{align*}
u&=u_0+u_1 A+u_2 A^2=u_1 A+u_2 A^2 \text{ (since $\Trmn(u)=\Trmn(A)=0$, then $u_0=0$)},\\
\frac1{1+c}&=\alpha_0+\alpha_1 A+\alpha_2 A^2,\text{ for some $u_i,\alpha_i, i\in\{0,1,2\}$ in $\F_{2^m}$}.
\end{align*}

We now go back to Equation~\eqref{eq:Lu}. We obtain (we use below that $\Trmn(A)=0, \Trmn(A^2)=(\Trmn(A))^2=0,\Trmn(A^4)=(\Trmn(A))^4=0$),
\allowdisplaybreaks
\begin{align*}
L(u)&=\dfrac{u_1 A+u_2 A^2}{\delta'}+(u_1 A+u_2 A^2)^2\\
&\qquad +\dfrac{A^2}{\delta'^2}\,\Trmn\left( \left(u_1 A+u_2 A^2\right) \left( \alpha_0+\alpha_1 A+\alpha_2 A^2\right)\right)^2\\
&= \gamma_0 \gamma_2 u_2^2 + A (\dfrac{u_1}{\delta'} + \gamma_0 u_2^2 + \gamma_1 \gamma_2 u_2^2) +  A^2 (u_1^2 + \dfrac{u_2}{\delta'} + \gamma_1 u_2^2 + \gamma_2^2 u_2^2)\\
&\qquad + \dfrac{A^2}{\delta'^2}\,\Trmn(\alpha_2  \gamma_0 u_1 + \alpha_1  \gamma_0 u_2 + \alpha_2  \gamma_0  \gamma_2 u_2 \\
&\qquad + 
 A \left(\alpha_0 u_1 + \alpha_2  \gamma_1 u_1 + \alpha_1  \gamma_1 u_2 + \alpha_2 ( \gamma_0 +  \gamma_1  \gamma_2) u_2\right)\\
 & \qquad + 
 A^2 (\alpha_1 u_1 + \alpha_2  \gamma_2 u_1 + \alpha_0 u_2 + \alpha_1  \gamma_2 u_2 + \alpha_2 ( \gamma_1 +  \gamma_2^2) u_2))^2\\
 &=\gamma_0 \gamma_2 u_2^2 + A (\dfrac{u_1}{\delta'} + \gamma_0 u_2^2 + \gamma_1 \gamma_2 u_2^2)\\
 & \qquad   +  A^2 (u_1^2 + \dfrac{u_2}{\delta'} + \gamma_1 u_2^2 + \gamma_2^2 u_2^2+\dfrac{\alpha_2  \gamma_0 u_1 + \alpha_1  \gamma_0 u_2 + \alpha_2  \gamma_0  \gamma_2 u_2}{\delta'^2}).
\end{align*}
Thus, $u$ is a solution if 
\allowdisplaybreaks
\begin{align*}
0&=\gamma_0 \gamma_2 u_2^2,\\
0&=\dfrac{u_1}{\delta'}+ \gamma_0 u_2^2 + \gamma_1 \gamma_2 u_2^2,\\
0&=u_1^2 + \dfrac{u_2}{\delta'} + \gamma_1 u_2^2 + \gamma_2^2 u_2^2+\dfrac{\alpha_2  \gamma_0 u_1 + \alpha_1  \gamma_0 u_2 + \alpha_2  \gamma_0  \gamma_2 u_2}{\delta'^2}.
\end{align*}
If $u_2=0$, then the system becomes
\allowdisplaybreaks
\begin{align*}
0&=u_2,\\
0&=\dfrac{u_1}{\delta'} ,\\
0&=u_1^2+\dfrac{\alpha_2\gamma_0 u_1}{\delta'^2},
\end{align*}
which implies $u=0$.

We now take $u_2\neq 0$, $\gamma_2=0$  (recall that $\gamma_0\neq 0$), and the system becomes 
\allowdisplaybreaks
\begin{align*}
0&=\gamma_2,\\
0&=\dfrac{u_1}{\delta'} + \gamma_0 u_2^2  \\
0&=u_1^2 + \dfrac{u_2}{\delta'} + \gamma_1 u_2^2 + \dfrac{\alpha_2  \gamma_0 u_1 + \alpha_1  \gamma_0 u_2}{\delta'^2} = u_1^2 + \gamma_1 u_2^2 +\dfrac{ \alpha_2  \gamma_0 u_1}{\delta'^2} +\left(\frac{1}{\delta'} + \dfrac{\alpha_1\gamma_0}{\delta'^2} \right)u_2.
\end{align*}
If $\alpha_1\gamma_0=\delta'$, that is, $\gamma_0=\dfrac{\delta'}{\alpha_1 }$, then $u_2=0$, or $u_2=\dfrac{\alpha_1\tilde\gamma_1 + \tilde \delta'\tilde \alpha_2
}{ \delta'^2}$, and $u_1=\dfrac{\alpha_1^2\gamma_1 +  \delta' \alpha_2
}{\alpha_1\delta'^2}$, where $\tilde \alpha_2^2 = \alpha_2$, $\tilde\gamma_1^2 = \gamma_1$ and $\tilde\delta'^2=\delta$ . 

If $\alpha_1\gamma_0\neq \delta'$, then we find 
$
u_1=\delta' \gamma_0 {u_2}^2,
$
which replaced in the third displayed equation renders
\[
(\gamma_0 \delta' {u_2}^2)^2 +\left(\gamma_1+\dfrac{\alpha_2 \gamma_0^2}{\delta'}\right)u_2^2+ \left(\dfrac{1}{\delta'}+\dfrac{\alpha_1 \gamma_0}{\delta'^2}\right)u_2 =0.
\]
Thus, since $u_2\neq 0$, since otherwise  $u_1=0$, so $u=0$, a case that we dealt with, dividing by $u_2$, we obtain
\[
 {u_2}^3 +\left(\dfrac{\gamma_1}{\gamma_0^2\delta'^2}+\dfrac{\alpha_2}{\delta'^3}\right)u_2+ \left(\dfrac{1}{\gamma_0^2\delta'^3}+\dfrac{\alpha_1 }{\gamma_0\delta'^4}\right) =0.
\] 
This equation over $\F_{2^m}^*$ has at most three solutions $u_2$ ($u_1$ is uniquely defined in terms of $u_2$). More precisely, when $\alpha_1\gamma_0\neq \delta'$, using the notations $b_0= \left(\dfrac{1}{\gamma_0^2\delta'^3}+\dfrac{\alpha_1 }{\gamma_0\delta'^4}\right)$, $b_1=\left(\dfrac{\gamma_1}{\gamma_0^2\delta'^2}+\dfrac{\alpha_2}{\delta'^3}\right)$, we know that this last equation has three solutions if and only if $\Tr_1^m (b_1^3/b_0^2)=\Tr_1^m(1)=m\pmod 2$ and $t_1,t_2$ are cubes in $\F_{2^m}$ if $m$ is even, and in $\F_{2^{2m}}$ if $m$ is odd, where $t_1,t_2$ are roots of $t^2+b_0t+b_1^3=0$.
Computationally, it seems that we always can get these conditions to be satisfied, but regardless, we get our upper bound.

To conclude, there are at most four solutions $u$ with $\Trmn(u)=0$ for Equation~\eqref{eq:Lu} and therefore for Equation~\eqref{eq:case3}.
\end{proof}

 \begin{thm} 
 \label{T1}
  Let $F(X)=(X^{2^m}+X+\delta)^{2^{2m}+1}+X$ over $\F_{2^{n}}$, where $n=3m$ and $\delta\in\F_{2^n}$. Let $\Gamma_1 :=\{\delta \in \F_{2^n}: \Trmn (\delta)=1\}$. Then$:$
  \begin{enumerate}
   \item[$(1)$] $F$ is  P$c$N for all $c \in \F_{2^m} \setminus \{1\}$ and for all $\delta\in\F_{2^n};$
  \item[$(2)$] $F$ is AP$c$N for all $c \in \F_{2^n}\setminus\F_{2^m}$ and for all $\delta \in \Gamma_1;$
   \item[$(3)$] $F$ is of $c$-differential uniformity $\leq 4$ for all $c \in \F_{2^n}\setminus\F_{2^m}$ and for all $\delta \in \F_{2^n} \setminus \Gamma_1$.
   \end{enumerate}
  \end{thm}
  \begin{proof} 
Clearly, by expanding the trinomial
\[
F(X)=X^{2^{2m}+2^m}+X^{2^{2m}+1}+X^{2^m+1}+\delta^{2^{2m}} X^{2^m}+ \delta X^{2^{2m}}+X^2+(\delta^{2^{2m}}+\delta+1)X + \delta^{2^{2m}+1}.
\]
Recall that, for any $(a, b) \in \F_{2^{n}} \times \F_{2^{n}}$,
the $c$-DDT entry $_c\Delta_F(a, b)$ is given by the number of solutions $X \in \F_{2^{n}}$ of the following equation,
\begin{equation}\label{eq1}
 F(X+a)+cF(X)=b,
\end{equation}
or, equivalently,
\begin{equation*}
 (1+c)F(X)+ (a^{2^{2m}} +a^{2^m})X+ (a^{2^{2m}} +a)X^{2^m}+ (a +a^{2^m})X^{2^{2m}} = b+F(a)+ \delta^{2^{2m}+1},\text{ that is,}
\end{equation*}
 \begin{equation*}
 (1+c)F(X)+ (a^{2^{2m}} +a^{2^m})X+ ((a^{2^{2m}} +a^{2^m})X)^{2^m}+ ((a^{2^{2m}} +a^{2^m})X)^{2^{2m}} = b+F(a)+\delta^{2^{2m}+1}.
\end{equation*}
This is the same as
 \begin{equation}
 \label{eq2}
 (1+c)F(X)+ \Trmn((a^{2^{2m}} +a^{2^m})X) = b+F(a)+ \delta^{2^{2m}+1}.
\end{equation}
Now, by using Equation~\eqref{ddtw}, the number of solutions $X \in \F_{2^n}$ of the above Equation~\eqref{eq2}, $_c\Delta_F(a,b)$, is given by 
$$\displaystyle{\frac{1}{2^n} \sum_{\beta \in \F_{2^n}} \sum_{X \in \F_{2^n}} (-1)^{\displaystyle{\Tr(\beta((1+c)F(X)+ \Trmn((a^{2^{2m}} +a^{2^m})X) + b+F(a)+ \delta^{2^{2m}+1}))}}},$$
or, equivalently,
\begin{equation*}
 \begin{split}
  _c\Delta_F(a,b) & = \dfrac{1}{2^n} \sum_{\beta \in \F_{2^n}} (-1)^{\displaystyle{\Tr\left(\beta \left(F(a)+b+\delta^{2^{2m}+1}\right)\right)}}\\
& \sum_{X \in \F_{2^n}} (-1)^{\displaystyle{\Tr\left(\beta \left((1+c)F(X)+ \Trmn\left((a^{2^{2m}} +a^{2^m})X\right)\right)\right)}}.
 \end{split}
\end{equation*}
Let $T_0 = \Tr(\beta (1+c)F(X))$ and $T_1 = \Tr(\beta (\Trmn((a^{2^{2m}} +a^{2^m})X)))$. Then the above equation becomes
\begin{equation}
\label{eq3}
 _c\Delta_F(a,b) = \frac{1}{2^n} \sum_{\beta \in \F_{2^n}} (-1)^{\displaystyle{\Tr\left(\beta \left(F(a)+b+\delta^{2^{2m}+1}\right)\right)}}\sum_{X \in \F_{2^n}} (-1)^{\displaystyle{T_0+T_1}}.
\end{equation}

\noindent
\textbf{Case 1.} Let $c\in \F_{2^m}\setminus \{1\}$ and $\delta \in \F_{2^n}$. To compute $T_0$ and $T_1$, we first write
\allowdisplaybreaks
\begin{equation*} 
\begin{split}
T_{1} & =  \Tr\left(\beta\Trmn((a^{2^{2m}} +a^{2^m})X)\right) \\
 & = \Tr\left(\beta((a^{2^{2m}} +a^{2^m})X + (a + a^{2^{2m}})X^{2^m} + (a^{2^m}+a)X^{2^{2m}})\right) \\
 & = \Tr\left((a^{2^{2m}} +a^{2^m})\Trmn(\beta)X\right),
\end{split}
\end{equation*}
and
\allowdisplaybreaks
\begin{align*} 
T_{0} & =  \Tr(\beta(1+c)F(X)) \\
 & = \Tr\left(\beta(1+c)(X^{2^{2m}+2^m}+X^{2^{2m}+1}+X^{2^m+1}+\delta^{2^{2m}} X^{2^m}+ \delta X^{2^{2m}}\right.\\
 &\qquad\qquad \left.+X^2+(\delta^{2^{2m}}+\delta+1)X + \delta^{2^{2m}+1})\right)\\
 & = \Tr\left(\beta(1+c) \delta^{2^{2m}+1}\right) + \Tr\left(\beta(1+c)(X^{2^{2m}+2^m}+X^{2^{2m}+1}+X^{2^m+1})\right) \\
 &\qquad\qquad + \Tr\left(\beta(1+c)(\delta^{2^{2m}} X^{2^m}+ \delta X^{2^{2m}}+X^2+(\delta^{2^{2m}}+\delta+1)X)\right)   \\
 & = \Tr\left(\beta(1+c) \delta^{2^{2m}+1}\right)  + \Tr\left((1+c) ( \beta + \beta^{2^m} +\beta^{2^{2m}})X^{2^m+1}\right)\\ 
  &\qquad\qquad + \Tr\left(\beta(1+c)(\delta^{2^{2m}} X^{2^m}+ \delta X^{2^{2m}}+X^2+(\delta^{2^{2m}}+\delta+1)X)\right)\\
  &=\Tr\left(\beta(1+c) \delta^{2^{2m}+1}\right) + \Tr\left((1+c)\Trmn(\beta) X^{2^m+1}\right)\\
  &\qquad\qquad  + \Tr\left( (1+c) \delta^{2^m} \Trmn(\beta) X+ \left(((1+c)\beta)^{2^{3m-1}}+\beta (1+c)(1+\Trmn(\delta))\right) X\right).
\end{align*}

Now Equation~\eqref{eq3} reduces to
\begin{equation*}
 _c\Delta_F(a,b) = \frac{1}{2^n} \sum_{\beta \in \F_{2^n}} (-1)^{\displaystyle{\Tr(\beta(F(a)+b+c\delta^{2^{2m}+1})}}\sum_{X \in \F_{2^n}} (-1)^{\displaystyle{\Tr(uX^{2^m+1}+vX)}},
\end{equation*}
where $u = (1+c)\Trmn(\beta)$ and $v= ((1+c)\delta^{2^m}  +a^{2^{2m}} +a^{2^m})\Trmn(\beta)+ ((1+c)\beta)^{2^{3m-1}}+ \beta(1+c)(1+\Trmn(\delta)).$ 

Further, splitting the above sum depending on whether $\Trmn(\beta)$ is~$0$ or not, we get
\allowdisplaybreaks
\begin{align*}
 2^n\, _c\Delta_F(a,b) & = \sum_{\substack{\beta \in \F_{2^n} \\ \Trmn(\beta)=0}} (-1)^{\displaystyle{\Tr(\beta(F(a)+b+c\delta^{2^{2m}+1}))}}\\
 &\qquad\qquad\sum_{X \in \F_{2^n}} (-1)^{\displaystyle{\Tr\left(\left(\left((1+c)\beta\right)^{2^{3m-1}}+(1+c)(1+\Trmn(\delta))\beta\right)X\right)}} \\
 & \quad +  \sum_{\substack{\beta \in \F_{2^n} \\ \Trmn(\beta)\neq 0}} (-1)^{\displaystyle{\Tr(\beta(F(a)+b+c\delta^{2^{2m}+1}))}}\sum_{X \in \F_{2^n}} (-1)^{\displaystyle{\Tr(uX^{2^m+1}+vX)}}\\
 &= S_0 + S_1,
\end{align*}
where $S_0,S_1$ are the two inner sums.

Now, to compute $S_0$, we write
\allowdisplaybreaks
\begin{align*}
 S_0 & = \sum_{\substack{\beta \in \F_{2^n} \\ \Trmn(\beta)=0}} (-1)^{\displaystyle{\Tr(\beta(F(a)+b+c\delta^{2^{2m}+1}))}}\\
 & \quad\quad\quad \sum_{X \in \F_{2^n}} (-1)^{\displaystyle{\Tr\left(\left(\left((1+c)\beta\right)^{2^{3m-1}}+(1+c)(1+\Trmn(\delta))\beta\right)X\right)}} \\
 & = 2^n + \sum_{\substack{{\beta \in \F_{2^n}^{*}}\\ \Trmn(\beta)=0}} (-1)^{\displaystyle{\Tr(\beta(F(a)+b+c\delta^{2^{2m}+1}))}}\\
 &\qquad\qquad\qquad \sum_{X \in \F_{2^n}} (-1)^{\displaystyle{\Tr((((1+c)\beta)^{2^{3m-1}}+(1+c)(1+\Trmn(\delta))\beta)X)}} \\
 & = 2^n.
\end{align*}
The reason for the above sum $S_0$ being $2^n$ is as follows. Since $((1+c)\beta)^{2^{3m-1}}+(1+c)(1+\Trmn(\delta))\beta)X$ is a permutation over $\F_{2^n}$, when $\Trmn(\delta)=1$, thus making the inner sum vanish. When $\Trmn(\delta) \neq 1$, then $((1+c)\beta)^{2^{3m-1}}+(1+c)(1+\Trmn(\delta))\beta$ vanishes for two values of $\beta$, namely, $\beta=0$ and $\beta =\dfrac{1}{(1+c)(1+\Trmn(\delta)^2)}$. Since for the later $\beta$, $\Trmn (\beta)=\Trmn\left(\dfrac{1}{(1+c)(1+\Trmn(\delta)^2)}\right) \neq 0$, we can exclude this $\beta$ from the inner sum of $S_0$. 

Next,
\allowdisplaybreaks
\begin{align*}
 S_1 & = \sum_{\substack{\beta \in \F_{2^n} \\ \Trmn(\beta)\neq 0}} (-1)^{\displaystyle{\Tr(\beta(F(a)+b+c\delta^{2^{2m}+1}}))}\sum_{X \in \F_{2^n}} (-1)^{\displaystyle{\Tr(uX^{2^m+1}+vX)}}\\
 & = \sum_{\substack{\beta \in \F_{2^n} \\ \Trmn(\beta)\neq 0}} (-1)^{\displaystyle{\Tr(\beta(F(a)+b+c\delta^{2^{2m}+1}}))}\mathcal{W}_G(v),
\end{align*}
where $\mathcal{W}_G(v)$ is the Walsh coefficient of the trace of the function $G : X \mapsto u X^{2^m+1}$. From Lemma~\ref{lemma1S1}, we know that $\mathcal{W}_G(v)=0$ if $\Tr(v)=0$. It is easy to see that 
 \begin{align*}
 \Tr(v)=&  \Tr\left( ((1+c)\delta^{2^m}  +a^{2^{2m}} +a^{2^m})\Trmn(\beta)+ ((1+c)\beta)^{2^{3m-1}}+\beta(1+c)(1+\Trmn(\delta))\right)\\
 =&0.
  \end{align*}
Hence the claim is shown.

\noindent
\textbf{Case 2.}
 Let $c\in\F_{2^n} \setminus \F_{2^m}$ and $\delta \in \Gamma_1$. Due to the restriction on $c$, the expression for $T_0$ will change as follows, while the expression for $T_1$ will stay invariant. Thus,
 \allowdisplaybreaks
\begin{align*} 
T_{0} & =  \Tr(\beta(1+c)F(X)) \\
 & = \Tr\left(\beta(1+c) \delta^{2^{2m}+1}\right)  + \Tr\left((1+c) ( \beta + \beta^{2^m} +\beta^{2^{2m}})X^{2^m+1}\right)\\ 
  &\qquad\qquad + \Tr\left(\beta(1+c)(\delta^{2^{2m}} X^{2^m}+ \delta X^{2^{2m}}+X^2+(\delta^{2^{2m}}+\delta+1)X)\right)\\
  &=\Tr\left(\beta(1+c) \delta^{2^{2m}+1}\right) + \Tr\left(\Trmn((1+c)\beta) X^{2^m+1}\right)\\
  &\qquad\qquad  + \Tr\left(  \delta^{2^m} \Trmn((1+c)\beta) X+ \left(((1+c)\beta)^{2^{3m-1}}+\beta (1+c)(1+\Trmn(\delta))\right) X\right).
\end{align*}

Now Equation~\eqref{eq3} reduces to
\begin{equation*}
 _c\Delta_F(a,b) = \frac{1}{2^n} \sum_{\beta \in \F_{2^n}} (-1)^{\displaystyle{\displaystyle \Tr(\beta(F(a)+b+c\delta^{2^{2m}+1})}}\sum_{X \in \F_{2^n}} (-1)^{\displaystyle{\Tr(uX^{2^m+1}+vX)}},
\end{equation*}
where 
\allowdisplaybreaks
\begin{align*}
u &=\Trmn((1+c)\beta),\\
v&=\delta^{2^m}\Trmn((1+c)\beta) + ((1+c)\beta)^{2^{3m-1}}+\beta(1+c)(1+\Trmn(\delta)) +(a^{2^{2m}} +a^{2^m}) \Trmn(\beta)\\
&=\delta^{2^m}\Trmn((1+c)\beta) + ((1+c)\beta)^{2^{3m-1}}+(a^{2^{2m}} +a^{2^m}) \Trmn(\beta).
\end{align*}
Now, we distribute the above sum in two parts depending on whether $\Trmn((1+c)\beta)=0$, or $\Trmn((1+c)\beta) \neq 0$, and we get
\allowdisplaybreaks
\begin{align*}
 2^n\, _c\Delta_F(a,b)  = & \sum_{\substack{\beta \in \F_{2^n} \\ \Trmn((1+c)\beta)=0}} (-1)^{\Tr(\beta(F(a)+b+c\delta^{2^{2m}+1}))} \\
 & \sum_{X \in \F_{2^n}} (-1)^{\Tr\left({\displaystyle{(((1+c)\beta)^{2^{3m-1}} + (a^{2^{2m}} +a^{2^m}) \Trmn(\beta))X}}\right)} \\
  & +   \sum_{\substack{\beta \in \F_{2^n} \\ \Trmn((1+c)\beta)\neq 0}} (-1)^{\Tr(\beta(F(a)+b+c\delta^{2^{2m}+1}))}\sum_{X \in \F_{2^n}} (-1)^{\displaystyle{\Tr(uX^{2^m+1}+vX)}} \\
 = & S_0 + S_1,
\end{align*}
where
\allowdisplaybreaks
\begin{align*}
S_0 = & \displaystyle \sum_{\substack{\beta \in \F_{2^n} \\ \Trmn((1+c)\beta)=0}}(-1)^{\displaystyle{\Tr(\beta(F(a)+b+c\delta^{2^{2m}+1}))}} \\
&\sum_{X \in \F_{2^n}} (-1)^{\displaystyle{\Tr((((1+c)\beta)^{2^{3m-1}}+(a^{2^{2m}} +a^{2^m}) \Trmn(\beta))X)}},\\
S_1 = & \displaystyle \sum_{\substack{\beta \in \F_{2^n} \\ \Trmn((1+c)\beta)\neq 0}} (-1)^{\displaystyle{\Tr(\beta(F(a)+b+c\delta^{2^{2m}+1}))}}\sum_{X \in \F_{2^n}} (-1)^{\displaystyle{\Tr(uX^{2^m+1}+vX)}}.
\end{align*}

We will now investigate two situations in order to calculate the sum $S_0$.  We observe that $S_0 = 2^n$ for those $a \in \F_{2^n}$ that satisfy $a^{2^{2m}} +a^{2^m}=0$. We will therefore assume that $a^{2^{2m}} +a^{2^m}\neq0$. We will try to find those $\beta$'s for which $(((1+c)\beta)^{2^{3m-1}} + (a^{2^{2m}} +a^{2^m}) \Trmn(\beta))X$ is a permutation; or equivalently, to find those  $\beta$'s for which $((1+c)\beta)^{2^{3m-1}} + (a^{2^{2m}} +a^{2^m}) \Trmn(\beta)\neq 0$, so that the inner sum in $S_0$ becomes zero.
If $H(\beta):=\dfrac{((1+c)\beta)^{2^{3m-1}}}{(a^{2^{2m}} +a^{2^m})} + \Trmn(\beta) = 0$, then we have,
$$H(\beta)+\Trmn(H(\beta))= \dfrac{((1+c)\beta)^{2^{3m-1}}}{(a^{2^{2m}} +a^{2^m})} +  \Trmn(\dfrac{((1+c)\beta)^{2^{3m-1}}}{(a^{2^{2m}} +a^{2^m})})=0. $$
With $z=\dfrac{((1+c)\beta)^{2^{3m-1}}}{(a^{2^{2m}} +a^{2^m})}$, we have $\Trmn(z)=z$, i.e., $z^{2^m}=z^{2^{2m}}$.

Also, $\beta = \dfrac{((a^{2^{2m}} +a^{2^m})z)^2}{1+c}$ and hence we have $$H(\beta)= z+ \dfrac{((a^{2^{2m}} +a^{2^m})z)^2}{1+c} + \dfrac{((a^{2^{2m}} +a^{2^m})z)^{2^{m+1}}}{(1+c)^{2^m}}+ \dfrac{((a^{2^{2m}} +a^{2^m})z)^{2^{2m+1}}}{(1+c)^{2^{2m}}}.$$ Because $z^{2^m}=z^{2^{2m}}$, the above equation further reduces to, 
\begin{equation}
\label{eqz}
 H(\beta)=z+ \Trmn \left( \dfrac{(a^{2^{2m}} +a^{2^m})^2}{1+c} \right) z^2 =0.
 \end{equation}

Note that if $\Trmn \left( \dfrac{(a^{2^{2m}} +a^{2^m})^2}{1+c} \right)=0$, Equation~\eqref{eqz} has a unique solution $z=0$. Consequently, $\beta=0$ is the only solution for $H(\beta)=0$, and in this case, it turns out that $S_0=2^n$.

If $\Trmn \left( \dfrac{(a^{2^{2m}} +a^{2^m})^2}{1+c} \right) \neq 0$, Equation~\eqref{eqz} has exactly two solutions, namely, $z_1=0$ and $z_2=\left(\Trmn \left (\dfrac{(a^{2^{2m}} +a^{2^m})^2}{1+c} \right)\right)^{-1}$. Equivalently, there are exactly two solutions for $H(\beta)=0$ given by $\beta_1 =0$ and $\beta_2 = \dfrac{(a^{2^{2m}} +a^{2^m})^2}{1+c} \left( \Trmn \left (\dfrac{(a^{2^{2m}} +a^{2^m})^2}{1+c} \right) \right)^{-2}$.

Thus for $a,b \in \F_{2^n}$, together with $a^{2^m}+a^{2^{2m}} \neq 0$, we have 
$$S_0 = 2^n\left(1 + (-1)^{\Tr(\beta_2 (F(a)+b+c\delta^{2^{2m}+1}))}\right).$$ 
Observe that if we take $b=c+F(a)$, then we have  $\Tr(\beta_2(F(a)+b+c)) =0$. Hence, $S_0=2^{n+1}$ for $(a,b) \in \F_{2^{n}} \times \F_{2^n}$, with $a^{2^m}+a^{2^{2m}} \neq 0$ and $b=c\delta^{2^{2m}+1}+F(a)$. Next,
\allowdisplaybreaks
\begin{align*}
 S_1 & = \sum_{\substack{\beta \in \F_{2^n} \\ \Trmn(\beta)\neq 0}} (-1)^{\displaystyle{\Tr(\beta(F(a)+b+c\delta^{2^{2m}+1}}))}\sum_{X \in \F_{2^n}} (-1)^{\displaystyle{\Tr(uX^{2^m+1}+vX)}}\\
 & = \sum_{\substack{\beta \in \F_{2^n} \\ \Trmn(\beta)\neq 0}} (-1)^{\displaystyle{\Tr(\beta(F(a)+b+c\delta^{2^{2m}+1}}))}\mathcal{W}_G(v),
\end{align*}
where $\mathcal{W}_G(v)$ is the Walsh coefficient of the trace of the function $G : X \mapsto u X^{2^m+1}$. Since $\Tr(v) = \Tr\left( \delta^{2^m}\Trmn((1+c)\beta) + ((1+c)\beta)^{2^{3m-1}} + (a^{2^{2m}} +a^{2^m}) \Trmn(\beta)\right)=0$ and $u \in \F_{2^m}$. Thus, by using Lemma~\ref{lemma1S1}, one can see that $S_1=0$. Hence, $F$ is AP$c$N in this case.

\noindent \textbf{Case 3.}  Let $c\in\F_{2^n} \setminus \F_{2^m}$ and $\delta \in \F_{2^n} \setminus \Gamma_1$.
Consider the following equation 
\begin{equation*}
 _c\Delta_F(a,b) = \frac{1}{2^n} \sum_{\beta \in \F_{2^n}} (-1)^{\displaystyle{\displaystyle \Tr(\beta(F(a)+b+c\delta^{2^{2m}+1})}}\sum_{X \in \F_{2^n}} (-1)^{\displaystyle{\Tr(uX^{2^m+1}+vX)}},
\end{equation*}
where 
\allowdisplaybreaks
\begin{align*}
u &=\Trmn((1+c)\beta),\\
v&=\delta^{2^m}\Trmn((1+c)\beta) + ((1+c)\beta)^{2^{3m-1}}+\beta(1+c)(1+\Trmn(\delta)) +(a^{2^{2m}} +a^{2^m}) \Trmn(\beta).
\end{align*}
Similar to the previous case, we split the above sum in two sums, namely, $S_0$ and $S_1$ depending upon $\Trmn((1+c)\beta)=0$ and $\Trmn((1+c)\beta) \neq 0$, respectively. 
Precisely,
\allowdisplaybreaks
\begin{align*}
S_0 = & \displaystyle \sum_{\substack{\beta \in \F_{2^n} \\ \Trmn((1+c)\beta)=0}}(-1)^{\displaystyle{\Tr(\beta(F(a)+b+c\delta^{2^{2m}+1}))}} \\
&\sum_{X \in \F_{2^n}} (-1)^{\displaystyle{\Tr((((1+c)\beta)^{2^{3m-1}}+\beta(1+c)(1+\Trmn(\delta))+(a^{2^{2m}} +a^{2^m}) \Trmn(\beta))X)}},\\
\text{and}\\
S_1 = & \displaystyle \sum_{\substack{\beta \in \F_{2^n} \\ \Trmn((1+c)\beta)\neq 0}} (-1)^{\displaystyle{\Tr(\beta(F(a)+b+c\delta^{2^{2m}+1}))}}\sum_{X \in \F_{2^n}} (-1)^{\displaystyle{\Tr(uX^{2^m+1}+vX)}}.
\end{align*}
It is clear from Lemma~\ref{lemma2S1} that the following equation
$$((1+c)\beta)^{2^{3m-1}}+\beta(1+c)(1+\Trmn(\delta))+(a^{2^{2m}} +a^{2^m}) \Trmn(\beta)=0$$ has at most four solutions in $\F_{2^n}$ and as a consequence, the maximum value that $S_0$ can attain is $2^{n+2}$.

Now, we consider $S_1$

\begin{align*}
 S_1 = & \displaystyle \sum_{\substack{\beta \in \F_{2^n} \\ \Trmn((1+c)\beta)\neq 0}} (-1)^{\displaystyle{\Tr(\beta(F(a)+b+c\delta^{2^{2m}+1}))}}\mathcal{W}_G(v),
\end{align*}
where $\mathcal{W}_G(v)$ is the Walsh coefficient of the trace of the function $G : X \mapsto u X^{2^m+1}$.

Since
\begin{align*}
 \Tr(v) & = \Tr\left( \delta^{2^m}\Trmn((1+c)\beta) + ((1+c)\beta)^{2^{3m-1}}\right) \\
  &   \qquad \qquad + \Tr\left(\beta(1+c)(1+\Trmn(\delta))+(a^{2^{2m}} +a^{2^m}) \Trmn(\beta)\right), \\
  & = \Tr\left( \delta^{2^m}(1+c)\beta+(\delta(1+c)\beta)^{2^m} +\delta^{2^m}((1+c)\beta)^{2^{2m}}+ ((1+c)\beta)^{2^{3m-1}}\right)\\
  & \qquad \qquad + \Tr\left(\beta(1+c)+ \delta\beta(1+c)+\delta^{2^m}\beta(1+c)+\delta^{2^{2m}}\beta(1+c) \right) =0,
\end{align*}
and $u \in \F_{2^m}$. It follows from Lemma~\ref{lemma1S1} that $S_1=0$.
\end{proof}

Next, we consider the permutation polynomial $F(X)=(X^{2^m}+X+\delta)^{2^{2m-1}+2^{m-1}}+X$ over  $\F_{2^{n}}$, where $n=3m, \delta \in \F_{2^n}$ and $m \not\equiv1 \pmod 3$. This is actually obtained by setting $i=2$ in part 2 of Lemma~\ref{L00}. The following theorem discusses the $c$-differential uniformity of the permutation $F$ depending on where the values of $c$ and $\delta$ lie.
\begin{thm} 
\label{T2}
  Let $F(X)=(X^{2^m}+X+\delta)^{2^{2m-1}+2^{m-1}}+X$ over $\F_{2^{n}}$, where $n=3m, \delta\in\F_{2^n}$ and $m \not\equiv1 \pmod 3$. Let $\Gamma_0 :=\{\delta \in \F_{2^n}: \Trmn (\delta)=0\}$. Then$:$
  \begin{enumerate}
   \item[$(1)$] $F$ is  P$c$N for all $c \in \F_{2^m} \setminus \{1\}$ and for all $\delta\in\F_{2^n};$
  \item[$(2)$] $F$ is AP$c$N for all $c \in \F_{2^n}\setminus\F_{2^m}$ and for all $\delta \in \Gamma_0;$
   \item[$(3)$] $F$ is of $c$-differential uniformity $\leq 4$ for all $c \in \F_{2^n}\setminus\F_{2^m}$ and for all $\delta \in \F_{2^n} \setminus \Gamma_0$.
   \end{enumerate}
  \end{thm}
\begin{proof} 
 Clearly, $F(X)=X^{2^{3m-1}+2^{2m-1}}+X^{2^{2m-1}+2^{m-1}}+X^{2^{m-1}+2^{3m-1}}+\delta^{2^{m-1}}X^{2^{2m-1}}+\delta^{2^{m-1}}X^{2^{3m-1}}+ X^{2^{2m}}+\delta^{2^{2m-1}}X^{2^{m-1}}+\delta^{2^{2m-1}}X^{2^{2m-1}}+X +\delta^{2^{m-1} +2^{2m-1}} = \Tr(X^{2^{2m-1}+2^{m-1}})+ X^{2^{2m}}+\delta^{2^{m-1}}X^{2^{3m-1}}+(\delta^{2^{m-1}}+\delta^{2^{2m-1}})X^{2^{2m-1}}+\delta^{2^{2m-1}}X^{2^{m-1}}+X +\delta^{2^{m-1} +2^{2m-1}}$.\\
Recall that, for any $(a, b) \in \F_{2^{n}} \times \F_{2^{n}}$
the $c$-DDT entry $_c\Delta_F(a, b)$ is given by the number of solutions $X \in \F_{2^{n}}$ of the following equation.
\begin{equation*}
 F(X+a)+cF(X)=b,
\end{equation*}
or equivalently,
 \begin{equation}
 \label{eq3.1}
 (1+c)F(X)+ \Trmn((a^{2^{2m-1}} +a^{2^{3m-1}})X^{2^{m-1}}) = b+F(a)+ \delta^{2^{m-1} +2^{2m-1}}.
\end{equation}
It follows from Equation~\eqref{ddtw}, the number of solutions $X \in \F_{2^n}$ of the above equation is given by
 \allowdisplaybreaks
\begin{align*}
_c\Delta_F(a,b) &= \frac{1}{2^n} \sum_{\beta \in \F_{2^n}} \sum_{X \in \F_{2^n}} (-1)^{\displaystyle{\Tr\left(\beta((1+c)F(X)+ \Trmn((a^{2^{2m-1}} +a^{2^{3m-1}})X^{2^{m-1}})\right)}}\\
& \quad\quad\quad\quad(-1)^{\displaystyle{\Tr\left(F(a)+b+\delta^{2^{m-1} +2^{2m-1}}\right)}}\\
= & \frac{1}{2^n}  \sum_{\beta \in \F_{2^n}} (-1)^{\displaystyle{\Tr(\beta(F(a)+b+\delta^{2^{m-1} +2^{2m-1}}))}}   \\
 & \qquad  \sum_{X \in \F_{2^n}} (-1)^{\displaystyle{\Tr (\beta^2((1+c)^2F(X)^2+ \Trmn((a^{2^{2m}} +a)X^{2^m})))}} \\
  =& \frac{1}{2^n}  \sum_{\beta \in \F_{2^n}}(-1)^{\displaystyle{\Tr(\beta(F(a)+b+c\delta^{2^{m-1} +2^{2m-1}}))}}  \\
 & \sum_{X \in \F_{2^n}} (-1)^{\displaystyle{\Tr(\beta^2(1+c)^2 ((\Trmn( X^{2^m+1})+X^{2^{2m+1}}+(\delta^{2^m}+\delta^{2^{2m}})X^{2^{2m}}})} \\
 & \quad \quad \quad \quad (-1)^{\displaystyle{\Tr((\beta(1+c))^2(\delta^{2^{2m}} X^{2^m}+X^2+\delta^{2^m}X)+\beta^2 \Trmn((a^{2^{2m}} +a^{2^m})X))}} \\
 = & \frac{1}{2^n}  \sum_{\beta \in \F_{2^n}}(-1)^{\displaystyle{\Tr(\beta(F(a)+b+c\delta^{2^{m-1} +2^{2m-1}})}}  \sum_{X \in \F_{2^n}} (-1)^{\displaystyle{\Tr(uX^{2^m+1}+vX)}}.
\end{align*}
Here, $u = \Trmn(\beta^2(1+c)^2 )$ and $v= \beta(1+c)+\delta^{2^m} (\beta(1+c))^2+ (\beta(1+c))^{2^m}+(\delta^{2^{2m}}+\delta)(\beta(1+c))^{2^{m+1}}+\delta^{2^m}(\beta(1+c))^{2^{2m+1}}+(a^{2^m}+a^{2^{2m}})\Trmn(\beta^2)$. 

\noindent

\textbf{Case 1.} Let $c \in \F_{2^m} \setminus \{1\}$ and $\delta \in \F_{2^n}$. Then $u=(1+c)^2 \Trmn(\beta^2)$ and $v=(1+c)(\beta +\beta^{2^m}) +\delta^{2^m}(1+c)^2(\beta^2 + \beta^{2^{2m+1}})+(\delta^{2^{2m}}+\delta)(1+c)^2\beta^{2^{m+1}}+(a^{2^m}+a^{2^{2m}})\Trmn(\beta^2)$. Further, splitting the sum for $\Trmn(\beta)=0$ and $\Trmn(\beta)\neq0$, we have $S_0$ and $S_1$, defined as below:
 \allowdisplaybreaks
\begin{align*}
 S_0  = & \sum_{\substack{\beta \in \F_{2^n} \\ \Trmn(\beta)=0}} (-1)^{\displaystyle{\Tr(\beta(F(a)+b+c\delta^{2^{m-1} +2^{2m-1}}))}}
  \sum_{X \in \F_{2^n}} (-1)^{\displaystyle{\Tr((1+c)(\beta +\beta^{2^m})X)}}
  \\
  & \qquad (-1)^{\displaystyle{\Tr\left((\delta^{2^m}(1+c)^2(\beta^2 + \beta^{2^{2m+1}})+(\delta^{2^{2m}}+\delta)(1+c)^2\beta^{2^{m+1}})X\right)}}\\
  = &\sum_{\substack{\beta \in \F_{2^n} \\ \Trmn(\beta)=0}} (-1)^{\displaystyle{\Tr(\beta(F(a)+b+c\delta^{2^{m-1} +2^{2m-1}}))}}\\
  &\sum_{X \in \F_{2^n}} (-1)^{\displaystyle{\Tr(((1+c)(\beta +\beta^{2^m}) +\Trmn(\delta)(1+c)^2\beta^{2^{m+1}})X)}}\\
  = & \sum_{\substack{\beta \in \F_{2^n} \\ \Trmn(\beta)=0}} (-1)^{\displaystyle{\Tr(\beta(F(a)+b+c\delta^{2^{m-1} +2^{2m-1}}))}}\\
  &\sum_{X \in \F_{2^n}} (-1)^{\displaystyle{\Tr((1+c)(\beta^{2^{2m}} +\Trmn(\delta)(1+c)\beta^{2^{m+1}})X)}}
  =  2^n,
\end{align*}
since $(\Trmn(\delta)(1+c)\beta^{2^{m+1}}+\beta^{2^{2m}})X$ is a permutation of $\F_{2^n}$ as $m \not\equiv 1 \pmod 3$, and hence the inner sum in $S_0$ is zero except for $\beta=0$.
 \allowdisplaybreaks
\begin{align*}
 S_1 & = \sum_{\substack{\beta \in \F_{2^n} \\ \Trmn(\beta)\neq 0}} (-1)^{\displaystyle{\Tr(\beta(F(a)+b+c\delta^{2^{m-1} +2^{2m-1}}))}}\sum_{X \in \F_{2^n}} (-1)^{\displaystyle{\Tr(uX^{2^m+1}+vX)}}\\
 & = \sum_{\substack{\beta \in \F_{2^n} \\ \Trmn(\beta)\neq 0}} (-1)^{\displaystyle{\Tr(\beta(F(a)+b+c\delta^{2^{m-1} +2^{2m-1}}))}}\sum_{X \in \F_{2^n}} (-1)^{\displaystyle{\Tr(X^{2^m+1}+(\gamma^{-1})vX)}}\\
 & = \sum_{\substack{\beta \in \F_{2^n} \\ \Trmn(\beta)\neq 0}} (-1)^{\displaystyle{\Tr(\beta(F(a)+b+c\delta^{2^{m-1} +2^{2m-1}}))}}\mathcal{W}_G(\gamma^{-1}v),
\end{align*}
where $G : X \mapsto X^{2^m+1}$ and $\gamma = (1+c) \Trmn(\beta)$. Also,
\allowdisplaybreaks
\begin{align*}
\Tr(\gamma^{-1}v) = & \Tr\left(\frac{(1+c)(\beta +\beta^{2^m}) +\delta^{2^m}(1+c)^2(\beta^2 + \beta^{2^{2m+1}})}{\gamma}\right)\\
  &+ \Tr\left(\frac{(\delta^{2^{2m}}+\delta)(1+c)^2\beta^{2^{m+1}}+(a^{2^m}+a^{2^{2m}})\Trmn(\beta^2)}{\gamma}\right)
  = 0.
\end{align*}
By using the similar arguments as in Theorem~\ref{T1}, one can show that the Walsh coefficient of $\Tr(X^{2^m+1})$ at $\gamma^{-1}v$ is $0$, and hence the proof of this case is done.

\noindent
\textbf{Case 2.}  Let $c\in\F_{2^n} \setminus \F_{2^m}$ and $\delta \in \Gamma_0$. Then $_c\Delta_F(a,b)$ is given by
 \allowdisplaybreaks
\begin{align*}
_c\Delta_F(a,b) =& \frac{1}{2^n}  \sum_{\beta \in \F_{2^n}}(-1)^{\displaystyle{\Tr(\beta(F(a)+b+c\delta^{2^{m-1} +2^{2m-1}}))}}  \\
 & \sum_{X \in \F_{2^n}} (-1)^{\displaystyle{\Tr(\beta^2(1+c)^2 ((\Trmn( X^{2^m+1})+X^{2^{2m+1}}+(\delta^{2^m}+\delta^{2^{2m}})X^{2^{2m}}})} \\
 & \quad (-1)^{\displaystyle{\Tr((\beta(1+c))^2(\delta^{2^{2m}} X^{2^m}+X^2+\delta^{2^m}X)+\beta^2 \Trmn((a^{2^{2m}} +a^{2^m})X))}} \\
 = & \frac{1}{2^n}  \sum_{\beta \in \F_{2^n}}(-1)^{\displaystyle{\Tr(\beta(F(a)+b+c\delta^{2^{m-1} +2^{2m-1}})}}  \sum_{X \in \F_{2^n}} (-1)^{\displaystyle{\Tr(uX^{2^m+1}+vX)}},
\end{align*}
where $u = \Trmn(\beta^2(1+c)^2 )$ and $v= \beta(1+c)+ (\beta(1+c))^{2^m}+ \delta^{2^m}(\Trmn(\beta(1+c))^2)+(a^{2^m}+a^{2^{2m}})\Trmn(\beta^2)$. Further, splitting the sum for $\Trmn((1+c)\beta)=0$ and $\Trmn((1+c)\beta)\neq0$, we have $S_0$ and $S_1$, defined below,
 \allowdisplaybreaks
\begin{align*}
 S_0  = & \sum_{\substack{\beta \in \F_{2^n} \\ \Trmn((1+c)\beta)=0}} (-1)^{\displaystyle{\Tr(\beta(F(a)+b+c\delta^{2^{m-1} +2^{2m-1}}))}}\\
 & \sum_{X \in \F_{2^n}} (-1)^{\displaystyle{\Tr(\beta(1+c)+(\beta(1+c))^{2^m}+(a^{2^m}+a^{2^{2m}})\Trmn(\beta^2))X)}}\\
  = &\sum_{\substack{\beta \in \F_{2^n} \\ \Trmn((1+c)\beta)=0}} (-1)^{\displaystyle{\Tr(\beta(F(a)+b+c\delta^{2^{m-1} +2^{2m-1}}))}}\\
  & \sum_{X \in \F_{2^n}} (-1)^{\displaystyle{\Tr((((1+c)\beta)^{2^{2m}} +(a^{2^m}+a^{2^{2m}})\Trmn(\beta^2))X)}}.
\end{align*}
Now, for those $a \in \F_{2^n}$ that satisfy $a^{2^{2m}} +a^{2^m}=0$, it is easy to see that $S_0 = 2^n$. Thus, let us assume that $a^{2^{2m}} +a^{2^m}\neq0$.
If $H(\beta):=\dfrac{((1+c)\beta)^{2^{2m}}}{(a^{2^{2m}} +a^{2^m})} + \Trmn(\beta^2) = 0$, then we have,
$$H(\beta)+\Trmn(H(\beta))= \dfrac{((1+c)\beta)^{2^{2m}}}{(a^{2^{2m}} +a^{2^m})} +  \Trmn\left(\dfrac{((1+c)\beta)^{2^{2m}}}{(a^{2^{2m}} +a^{2^m})}\right)=0. $$
Letting $z=\dfrac{((1+c)\beta)^{2^{2m}}}{(a^{2^{2m}} +a^{2^m})}$, we have $\Trmn(z)=z$, i.e., $z^{2^m}=z^{2^{2m}}$. 

Using $\beta = \dfrac{((a^{2^{2m}} +a^{2^m})z)^{2^m}}{1+c}$ and the same technique as in the Case 2 of Theorem~\ref{T1}, we have $z_1=0$ and $z_2 = \left( \Trmn \left (\dfrac{(a^{2^{2m}} +a^{2^m})^{2^{m+1}}}{(1+c)^2} \right)\right)^{-1}$. Equivalently, there are exactly two solutions for $H(\beta)=0$ given by $\beta_1 =0$ and, \\
$\beta_2 = \dfrac{(a^{2^{2m}} +a^{2^m})^{2^m}}{1+c} \left( \Trmn \left (\dfrac{(a^{2^{2m}} +a^{2^m})^{2^{m+1}}}{(1+c)^2}\right)\right)^{-2^m}$. Thus for $a,b \in \F_{2^n}$, together with $a^{2^m}+a^{2^{2m}} \neq 0$, we have 
$$S_0  =  2^n\left(1 + (-1)^{\Tr(\beta_2 (F(a)+b+c\delta^{2^{m-1} +2^{2m-1}}))}\right).$$ 
Observe that if we take $b=c\delta^{2^{m-1} +2^{2m-1}}+F(a)$, then we have  $\Tr(\beta_2(F(a)+b+c)) =0$. Hence, $S_0=2^{n+1}$ for $(a,b) \in \F_{2^{n}} \times \F_{2^n}$, with $a^{2^m}+a^{2^{2m}} \neq 0$ and $b=c+F(a)$. Next, 
 \allowdisplaybreaks
\begin{align*}
 S_1 & = \sum_{\substack{\beta \in \F_{2^n} \\ \Trmn((1+c)\beta)\neq 0}} (-1)^{\displaystyle{\Tr(\beta(F(a)+b+c\delta^{2^{m-1} +2^{2m-1}}))}}\sum_{X \in \F_{2^n}} (-1)^{\displaystyle{\Tr(uX^{2^m+1}+vX)}}\\
 & = \sum_{\substack{\beta \in \F_{2^n} \\ \Trmn((1+c)\beta)\neq 0}} (-1)^{\displaystyle{\Tr(\beta(F(a)+b+c\delta^{2^{m-1} +2^{2m-1}}))}}\sum_{X \in \F_{2^n}} (-1)^{\displaystyle{\Tr(X^{2^m+1}+(\gamma^{-1})vX)}}\\
 & = \sum_{\substack{\beta \in \F_{2^n} \\ \Trmn((1+c)\beta)\neq 0}} (-1)^{\displaystyle{\Tr(\beta(F(a)+b+c\delta^{2^{m-1} +2^{2m-1}}))}}\mathcal{W}_G(\gamma^{-1}v),
\end{align*}
where $G : X \mapsto X^{2^m+1}$ and $\gamma = (1+c) \Trmn(\beta)$. Also, $ \Tr(\gamma^{-1}v) $ equals
$$
\Tr\left(\frac{\beta(1+c)+ (\beta(1+c))^{2^m}+ (\delta+\delta^{2^{2m}})(\Trmn(\beta(1+c))^2)+(a^{2^m}+a^{2^{2m}})\Trmn(\beta^2)}{\gamma}\right)=0.
$$
Hence, $\mathcal{W}_G(\gamma^{-1}v)=0$, thereby proving the claim.

\noindent
\textbf{Case 3.} Let $c\in\F_{2^n} \setminus \F_{2^m}$ and $\delta \in \F_{2^n} \setminus \Gamma_0$.
Then $_c\Delta_F(a,b)$ is given by
 \allowdisplaybreaks
\begin{align*}
_c\Delta_F(a,b) =& \frac{1}{2^n}  \sum_{\beta \in \F_{2^n}}(-1)^{\displaystyle{\Tr(\beta(F(a)+b+c\delta^{2^{m-1} +2^{2m-1}}))}}  \\
 & \sum_{X \in \F_{2^n}} (-1)^{\displaystyle{\Tr(\beta^2(1+c)^2 ((\Trmn( X^{2^m+1})+X^{2^{2m+1}}+(\delta^{2^m}+\delta^{2^{2m}})X^{2^{2m}}})} \\
 & \quad \quad  (-1)^{\displaystyle{\Tr((\beta(1+c))^2(\delta^{2^{2m}} X^{2^m}+X^2+\delta^{2^m}X)+\beta^2 \Trmn((a^{2^{2m}} +a^{2^m})X))}} \\
 = & \frac{1}{2^n}  \sum_{\beta \in \F_{2^n}}(-1)^{\displaystyle{\Tr(\beta(F(a)+b+c\delta^{2^{m-1} +2^{2m-1}})}}  \sum_{X \in \F_{2^n}} (-1)^{\displaystyle{\Tr(uX^{2^m+1}+vX)}},
\end{align*}
where 
 \allowdisplaybreaks
\begin{align*}
 u & = \Trmn(\beta^2(1+c)^2 ) \\
 v & = \beta(1+c)+ (\beta(1+c))^{2^m}+ \delta^{2^m}((1+c)\beta)^2+\delta^{2^m}((1+c)\beta)^{2^{2m+1}} \\
 & \qquad \qquad +(\delta^{2^{2m}}+\delta)((1+c)\beta)^{2^{m+1}}+(a^{2^m}+a^{2^{2m}})\Trmn(\beta^2).
\end{align*}
 Again, dividing the above sum in two parts, $S_0$ and $S_1$ depending upon $\Trmn((1+c)\beta)=0$ and $\Trmn((1+c)\beta)\neq0$, respectively.
 
First, we consider $S_0$ given by the following expression.
\allowdisplaybreaks
\begin{align*}
 S_0  = & \sum_{\substack{\beta \in \F_{2^n} \\ \Trmn((1+c)\beta)=0}} (-1)^{\displaystyle{\Tr(\beta(F(a)+b+c\delta^{2^{m-1} +2^{2m-1}}))}}\\
 & \sum_{X \in \F_{2^n}} (-1)^{\displaystyle{\Tr(\beta(1+c)+(\beta(1+c))^{2^m}+(a^{2^m}+a^{2^{2m}})\Trmn(\beta^2))X)}}\\
  = &\sum_{\substack{\beta \in \F_{2^n} \\ \Trmn((1+c)\beta)=0}} (-1)^{\displaystyle{\Tr(\beta(F(a)+b+c\delta^{2^{m-1} +2^{2m-1}}))}} \\
  & \sum_{X \in \F_{2^n}} (-1)^{\displaystyle{\Tr(((\beta(1+c))^{2^{2m}}+ \delta^{2^m}((1+c)\beta)^2+\delta^{2^m}((1+c)\beta)^{2^{2m+1}})X)}}\\
  &  (-1)^{\displaystyle{\Tr(((\delta^{2^{2m}}+\delta)((1+c)\beta)^{2^{m+1}}+(a^{2^m}+a^{2^{2m}})\Trmn(\beta^2))X)}}.
\end{align*}
Further to compute $S_0$, we need to determine the solutions of the following equation:
$$(\beta(1+c))^{2^{2m}}+ \Trmn(\delta)((1+c)\beta)^{2^{m+1}} +(a^{2^m}+a^{2^{2m}})\Trmn(\beta^2)=0.$$ A similar analysis as in the proof of Theorem~\ref{T1} works here, as well, rendering at most four solutions for the above equation.

Next, we consider $S_1$ whose expression is given below.
 \allowdisplaybreaks
\begin{align*}
 S_1 & = \sum_{\substack{\beta \in \F_{2^n} \\ \Trmn((1+c)\beta)\neq 0}} (-1)^{\displaystyle{\Tr(\beta(F(a)+b+c\delta^{2^{m-1} +2^{2m-1}}))}}\sum_{X \in \F_{2^n}} (-1)^{\displaystyle{\Tr(uX^{2^m+1}+vX)}}\\
 & = \sum_{\substack{\beta \in \F_{2^n} \\ \Trmn((1+c)\beta)\neq 0}} (-1)^{\displaystyle{\Tr(\beta(F(a)+b+c\delta^{2^{m-1} +2^{2m-1}}))}}\sum_{X \in \F_{2^n}} (-1)^{\displaystyle{\Tr(X^{2^m+1}+(\gamma^{-1})vX)}}\\
 & = \sum_{\substack{\beta \in \F_{2^n} \\ \Trmn((1+c)\beta)\neq 0}} (-1)^{\displaystyle{\Tr(\beta(F(a)+b+c\delta^{2^{m-1} +2^{2m-1}}))}}\mathcal{W}_G(\gamma^{-1}v),
\end{align*}
where $G : X \mapsto X^{2^m+1}$ and $\gamma = (1+c) \Trmn(\beta)$. Also, 
\begin{align*}
 \Tr(\gamma^{-1}v) & = \Tr\left(\frac{\beta(1+c)+ (\beta(1+c))^{2^m}+ \delta^{2^m}((1+c)\beta)^2+\delta^{2^m}((1+c)\beta)^{2^{2m+1}}}{\gamma}\right) \\
 &  \qquad \qquad  + \Tr\left(\frac{(\delta^{2^{2m}}+\delta)((1+c)\beta)^{2^{m+1}}+(a^{2^m}+a^{2^{2m}})\Trmn(\beta^2)}{\gamma}\right)=0.
\end{align*}
Clearly, $\mathcal{W}_G(\gamma^{-1}v)=0$, making $S_1=0$.  
\end{proof}
 
 Substituting $i=3$ in part 2 of Lemma~\ref{L00} yields the permutation polynomial $F(X)=(X^{2^m}+X+\delta)^{2^{3m-1}+2^{m-1}}+X$ over $\F_{2^{n}}$ over $\F_{2^{n}}$, where $n=3m, \delta\in\F_{2^n}$ and $2m \not\equiv1 \pmod 3$, and we explore the $c$-differential uniformity of this permutation in the following theorem.
 \begin{thm} \label{T3}
  Let $F(X)=(X^{2^m}+X+\delta)^{2^{3m-1}+2^{m-1}}+X$ over $\F_{2^{n}}$, where $n=3m, \delta\in\F_{2^n}$ and $2m \not\equiv1 \pmod 3$. Let $\Gamma_0 :=\{\delta \in \F_{2^n}: \Trmn (\delta)=0\}$. Then$:$
  \begin{enumerate}
   \item[$(1)$] $F$ is  P$c$N for all $c \in \F_{2^m} \setminus \{1\}$ and for all $\delta\in\F_{2^n};$
  \item[$(2)$] $F$ is AP$c$N for all $c \in \F_{2^n}\setminus\F_{2^m}$ and for all $\delta \in \Gamma_0;$
   \item[$(3)$] $F$ is of $c$-differential uniformity $\leq 4$ for all $c \in \F_{2^n}\setminus\F_{2^m}$ and for all $\delta \in \F_{2^n} \setminus \Gamma_0$.
   \end{enumerate}
  \end{thm}
\begin{proof}
  The proof proceeds in a similar manner as the one of Theorem~\ref{T2}. 
\end{proof}

\section{Permutations over $\F_{p^n}$ with low $c$-differential uniformity}\label{S4}
In the preceding section, we studied the $c$-differential uniformity of some classes of permutation polynomials over fields of even characteristic. However, in this section, we discuss the $c$-differential uniformity of permutations over fields of odd characteristic. In the first theorem of this section, we consider the $c$-differential uniformity of polynomial $F(X)=(X^{3^m}-X+\delta)^{3^{2m-1}+2\cdot3^{m-1}}+X$, which is a permutation, as stated in Lemma~\ref{L02}, over $\F_{3^{n}}$, where $\delta \in \F_{3^n}$ and $n=2m$. 
\begin{thm} \label{T11}
  Let $F(X)=(X^{3^m}-X+\delta)^{3^{2m-1}+2\cdot3^{m-1}}+X$ over $\F_{3^{n}}$, where $n=2m$ and $\delta\in\F_{3^n}$. Let $\Gamma_0 :=\{\delta \in \F_{2^n}: \Trnn (\delta)=0\}$. Then$:$
  \begin{enumerate}
   \item[$(1)$] $F$ is  P$c$N for all $c \in \F_{3^m} \setminus \{1\}$ and for all $\delta\in\F_{3^n};$
   \item[$(2)$] $F$ is P$c$N for all $c \in \F_{3^n}\setminus\F_{3^m}$ and for all $\delta \in \Gamma_0$. Moreover, it is of $c$-differential uniformity $3$ for all $c \in \F_{3^n}\setminus\F_{3^m}$ and for all $\delta \in \F_{3^n} \setminus \Gamma_0$.    
 \end{enumerate}
  \end{thm}
\begin{proof} 
Clearly, after simplifying
\begin{align*}
 F(X)= & X^{3^m}+(\delta^{2\cdot{3^{m-1}}}+\delta^{3^{m-1}+3^{2m-1}})(X^{3^{m-1}}-X^{3^{2m-1}})\\
 + & (\delta^{3^{m-1}}+\delta^{3^{2m-1}})(X^{2\cdot3^{m-1}}+X^{2\cdot3^{2m-1}} +X^{3^{m-1}+3^{2m-1}})+\delta^{3^{2m-1}+2\cdot3^{m-1}}.
\end{align*}

We know that for any $(a, b) \in \F_{3^{n}} \times \F_{3^{n}}$,
the $c$-DDT entry $_c\Delta_F(a, b)$ is given by the number of solutions $X \in \F_{3^{n}}$ of the following equation
\begin{equation}\label{eq31}
 F(X+a)-cF(X)=b,
\end{equation}
or, equivalently,
\begin{equation*}
 (1-c)F(X)+ \Trnn((\delta^{3^{m-1}}+\delta^{3^{2m-1}})(a^{3^{2m-1}}-a^{3^{m-1}})X^{3^{m-1}}) + F(a)- \delta^{3^{2m-1}+2\cdot3^{m-1}} = b.
 \end{equation*}
Now, by using Equation~\eqref{ddtw}, the number of solutions $X \in \F_{3^n}$ of the above Equation~\eqref{eq31}, $_c\Delta_F(a,b)$, is given by 
\allowdisplaybreaks
\begin{align*}
 & \displaystyle{\frac{1}{3^n} \sum_{\beta \in \F_{3^n}} \sum_{X \in \F_{3^n}} \omega^{\displaystyle{\Tr(\beta((1-c)F(X)+ \Trnn((\delta^{3^{m-1}}+\delta^{3^{2m-1}})(a^{3^{2m-1}}-a^{3^{m-1}})X^{3^{m-1}})}}} \\
 & \hspace{3cm} \omega^{\displaystyle{\Tr(\beta(F(a)- \delta^{3^{2m-1}+2\cdot3^{m-1}}-b))}},
\end{align*}
where $\omega =e^{2\pi i/3}$; or equivalently,
\allowdisplaybreaks
\begin{align*}
  _c\Delta_F(a,b) & = \dfrac{1}{3^n} \sum_{\beta \in \F_{3^n}} \omega^{\displaystyle{\Tr\left(\beta\left(F(a)-b- \delta^{3^{2m-1}+2\cdot3^{m-1}}\right)\right)}}\\
& \sum_{X \in \F_{3^n}} \omega^{\displaystyle{\Tr\left(\beta((1-c)F(X)+ \Trnn((\delta^{3^{m-1}}+\delta^{3^{2m-1}})(a^{3^{2m-1}}-a^{3^{m-1}})X^{3^{m-1}})\right)}}.
 \end{align*}
Let 
\allowdisplaybreaks
\begin{align*}
 T_0 & = \Tr(\beta(1-c)F(X)), \\
 T_1 & = \Tr(\beta \Trnn((\delta^{3^{m-1}}+\delta^{3^{2m-1}})(a^{3^{2m-1}}-a^{3^{m-1}})X^{3^{m-1}})).
\end{align*}
Then the above equation becomes
\begin{equation}
\label{eq35}
 _c\Delta_F(a,b) = \frac{1}{3^n} \sum_{\beta \in \F_{3^n}} \omega^{\displaystyle{\Tr\left(\beta\left(F(a)-b- \delta^{3^{2m-1}+2\cdot3^{m-1}}\right)\right)}}\sum_{X \in \F_{3^n}} \omega^{\displaystyle{T_0+T_1}}.
\end{equation}

\noindent
\textbf{Case 1.} Let $c\in \F_{3^m}\setminus \{1\}$ and $\delta \in \F_{3^n}$. To compute $T_0$ and $T_1$, we first write
\allowdisplaybreaks
\begin{equation*} 
\begin{split}
T_{1} & = \Tr(\beta \Trnn((\delta^{3^{m-1}}+\delta^{3^{2m-1}})(a^{3^{2m-1}}-a^{3^{m-1}})X^{3^{m-1}})) \\
 & = \Tr\left(\beta(\delta^{3^{m-1}}+\delta^{3^{2m-1}})((a^{3^{2m-1}} -a^{3^{m-1}})X^{3^{m-1}} + (a^{3^{m-1}} - a^{3^{2m-1}})X^{3^{2m-1}})\right) \\
 & = \Tr\left((\delta+\delta^{3^m})(a^{3^m}-a)\Trnn(\beta^3)X\right),
\end{split}
\end{equation*}
and
\allowdisplaybreaks
\begin{align*} 
T_{0} & =  \Tr(\beta(1-c)F(X)) \\
 & = \Tr\left(\beta(1-c)(X^{3^m}+(\delta^{2\cdot{3^{m-1}}}+\delta^{3^{m-1}+3^{2m-1}})(X^{3^{m-1}}-X^{3^{2m-1}})\right.\\
 &\qquad\qquad \left.+ (\delta^{3^{m-1}}+\delta^{3^{2m-1}})(X^{2\cdot3^{m-1}}+X^{2\cdot3^{2m-1}} +X^{3^{m-1}+3^{2m-1}})+ \delta^{3^{2m-1}+2\cdot3^{m-1}})\right)\\
 & = \Tr\left(\beta(1-c)\delta^{3^{2m-1}+2\cdot3^{m-1}}\right)+ \Tr\left(\beta^3(1-c)^3(X^{3^{m+1}}+(\delta^{2\cdot{3^m}}+\delta^{3^m+1})(X^{3^m}-X)\right.\\
 &\qquad\qquad\quad\quad\quad \left.+ (\delta^{3^{m}}+\delta)(X^{2\cdot3^{m}}+X^2 +X^{3^m+1}))\right)\\
 & = \Tr\left(\beta(1-c)\delta^{3^{2m-1}+2\cdot3^{m-1}}\right)+ \Tr\left((\beta(1-c))^{3^m}X+(\delta^2+\delta^{3^m+1})(\beta(1-c))^{3^{m+1}}X \right.\\
 & \left. -(\delta^{2\cdot3^m}+\delta^{3^m+1})(\beta(1-c))^3X+ (\beta(1-c))^3(\delta^{3^{m}}+\delta)(X^{2\cdot3^{m}}+X^2 -2X^{3^m+1})\right)\\
 & = \Tr\left(\beta(1-c)\delta^{3^{2m-1}+2\cdot3^{m-1}}\right)+ \Tr\left((\beta(1-c))^{3^m}X+(\delta^2+\delta^{3^m+1})(\beta(1-c))^{3^{m+1}}X \right.\\
 & \left. -(\delta^{2\cdot3^m}+\delta^{3^m+1})(\beta(1-c))^3X+ (\delta^{3^{m}}+\delta)(\Trnn(\beta(1-c))^3)(X^2-X^{3^m+1})\right).
\end{align*}

Now Equation~\eqref{eq35} reduces to
\begin{equation*}
 _c\Delta_F(a,b) = \frac{1}{3^n} \sum_{\beta \in \F_{3^n}} \omega^{\displaystyle{\Tr(\beta(F(a)-b- c\delta^{3^{2m-1}+2\cdot3^{m-1}})}}\sum_{X \in \F_{3^n}} \omega^{\displaystyle{\Tr(u(X^2-X^{3^m+1})+vX)}},
\end{equation*}
where,
\begin{align*}
 u & = (\delta+\delta^{3^m})(1-c)^3\Trnn(\beta^3) \\
 v & = (\delta^{3^m}+\delta)(a^{3^m}-a)\Trnn(\beta^3)+(1-c)\beta^{3^m}+(\delta^2+\delta^{3^m+1})(1-c)^3 \beta^{3^{m+1}}\\
 & \qquad \qquad  -(\delta^{2\cdot3^m}+\delta^{3^m+1})(\beta(1-c))^3.
\end{align*}

Further, splitting the above sum depending on whether $\Trnn(\beta)$ is~$0$ or not, we get
\allowdisplaybreaks
\begin{align*}
3^n\, _c\Delta_F(a,b) & =  \sum_{\substack{\beta \in \F_{3^n} \\ \Trnn(\beta)=0}} \omega^{\displaystyle{\Tr(\beta(F(a)-b- c\delta^{3^{2m-1}+2\cdot3^{m-1}}))}}\sum_{X \in \F_{3^n}} \omega^{\displaystyle{\Tr\left((1-c)\beta^{3^m}X\right)}} \\
 & \omega^{\displaystyle{\Tr\left(\left((\delta^2+\delta^{3^m+1})(1-c)^3 \beta^{3^{m+1}}-(\delta^{2\cdot3^m}+\delta^{3^m+1})(\beta(1-c))^3\right)X\right)}}  \\
 &   +  \sum_{\substack{\beta \in \F_{3^n} \\ \Trnn(\beta)\neq 0}} \omega^{\displaystyle{\Tr(\beta(F(a)-b- c\delta^{3^{2m-1}+2\cdot3^{m-1}}))}} \\
 &   \qquad \qquad \qquad \qquad \sum_{X \in \F_{3^n}} \omega^{\displaystyle{\Tr(u(X^2-X^{3^m+1})+vX)}} \\
 &=  S_0 + S_1,
\end{align*}
where $S_0,S_1$ are the two inner sums.

Now, to compute $S_0$, we write
\allowdisplaybreaks
\begin{align*}
 S_0 & = \sum_{\substack{\beta \in \F_{3^n} \\ \Trnn(\beta)=0}} \omega^{\displaystyle{\Tr(\beta(F(a)-b- c\delta^{3^{2m-1}+2\cdot3^{m-1}}))}}\\
 & \sum_{X \in \F_{3^n}} \omega^{\displaystyle{\Tr\left(\left((1-c)\beta^{3^m}+(1-c)^3(-\delta^2-2\delta^{3^m+1}-\delta^{2\cdot3^m})\beta^3\right)X\right)}} \\
 & = 3^n + \sum_{\substack{{\beta \in \F_{3^n}^{*}}\\ \Trnn(\beta)=0}} \omega^{\displaystyle{\Tr(\beta(F(a)-b- c\delta^{3^{2m-1}+2\cdot3^{m-1}}))}}\\
 &\qquad\qquad\qquad \sum_{X \in \F_{3^n}} \omega^{\displaystyle{\Tr\left(\left((1-c)\beta^{3^m}-(1-c)^3(\Trnn(\delta))^2\beta^3\right)X\right)}} \\
 & = 3^n.
\end{align*}
 The justification for above equality is  as follows. In order to compute $S_0$, we need to solve the following cubic equation (as $\Trnn(\beta)=0$),
 \begin{equation}\label{eq33}
  (1-c)^2(\Trnn(\delta))^2\beta^3+\beta=0.
 \end{equation}
 If $\Trnn(\delta)=0$, then we have only unique solution, that is $\beta=0$. Let us assume $\Trnn(\delta)\neq0$, then the cubic Equation~\eqref{eq33} has two solutions $\beta_1$ and $-\beta_1$ in $\F_{3^n}^{*}$ because  $\dfrac{-1}{(1-c)^2(\Trnn(\delta))^2}$ is a square in $\F_{3^n}^{*}$, or equivalently $-1$ is a square in $\F_{3^n}^{*}$. But one can easily verify $\beta_1$ does not satisfy the condition $\Trnn(\beta_1) = 0$. Hence $\beta=0$ is only possible solution for Equation ~\eqref{eq33}, making $S_0 = 3^n$. 

Next, we consider the sum $S_1$,
\begin{align*}
 S_1 & = \sum_{\substack{\beta \in \F_{3^n} \\ \Trnn(\beta)\neq 0}} \omega^{\displaystyle{\Tr(\beta(F(a)-b- c\delta^{3^{2m-1}+2\cdot3^{m-1}}))}} \sum_{X \in \F_{3^n}} \omega^{\displaystyle{\Tr\left(u(X^2-X^{3^m+1})+vX\right)}}\\
 & = \sum_{\substack{\beta \in \F_{3^n} \\ \Trnn(\beta)\neq0}} \omega^{\displaystyle{\Tr(\beta(F(a)-b- c\delta^{3^{2m-1}+2\cdot3^{m-1}}))}} \mathcal{W}_G(-v),
 \end{align*}
 where $\mathcal{W}_G(-v)$ is Walsh transform of trace of function $G : X \mapsto u(X^{3^0+1}-X^{3^m+1})$ at $-v$. Now, from Lemma~\ref{walsh}, the absolute square of Walsh transform coefficient of $G$ is given by
 \begin{equation*} \lvert \mathcal{W}_G(-v) \rvert^2 =
  \begin{cases}
   3^{n+\ell} &~\mbox{if}~G(X)+\Tr(vX)\equiv0~\mbox{on Ker}~(L),  \\
    0 &~\mbox{otherwise},
  \end{cases}
 \end{equation*}
 where $\ell$ is dimension of kernel of the linearized polynomial $L(X)=u(X^{3^m}-X)$, and hence Ker~$(L) =  \F_{3^m} $. It is easy to see that for those $\beta$'s for which $v \neq 0$, we have $\mathcal{W}_f(-v)=0$. Now, we shall argue below that those $\beta$'s for which $v=0$ can essentially be ignored. In order to do so we first express $v$ in a slightly different form as follows.
 \begin{align*}
  v & = (\delta^{3^m}+\delta)(a^{3^m}-a)\Trnn(\beta^3)+(1-c)\beta^{3^m}+(\delta^2+\delta^{3^m+1})(1-c)^3 \beta^{3^{m+1}}\\
  & \qquad \qquad -(\delta^{2\cdot3^m}+\delta^{3^m+1})(\beta(1-c))^3, \\
  & = ((\delta^{3^m}+\delta)(a^{3^m}-a)+(\delta^{2}+\delta^{3^m+1})(1-c)^3) \beta^{3^{m+1}}+(1-c)\beta^{3^m} \\
  & \qquad \qquad +((\delta^{3^m}+\delta)(a^{3^m}-a)-(\delta^{2\cdot3^m}+\delta^{3^m+1})(1-c)^3) \beta^3,\\
  & = -A^{3^m}  \beta^{3^{m+1}} + (1-c)\beta^{3^m} + A \beta^3,\\
  & = -(A\beta^3)^{3^m}+(1-c)\beta^{3^m} + A \beta^3,
 \end{align*}
 where $A= (\delta^{3^m}+\delta)(a^{3^m}-a)-(\delta^{2\cdot3^m}+\delta^{3^m+1})(1-c)^3$. If $A=0$, then $v=0$ holds only for $\beta=0$, but then $\Trnn(\beta)=0$ and hence we can ignore this case. Let us now assume that $A \neq 0$, then substituting $Z=A\beta^3$, we have the following expression for $v=0$.
 \begin{equation}\label{eq4}
  v= -Z^{3^m}+Z+\frac{(1-c)}{A^{3^{m-1}}} Z^{3^{m-1}}=0.
 \end{equation}
  Raising the above Equation~\eqref{eq4} to $3^m$ and adding to the original equation, we have $\dfrac{(1-c)}{A^{3^{m-1}}} Z^{3^{m-1}} + \dfrac{(1-c)}{A^{3^{2m-1}}} Z^{3^{2m-1}}=0$, or equivalently, $\dfrac{(1-c)^3}{A^{3^m}} Z^{3^m} + \dfrac{(1-c)^3}{A} {Z}=0$, i.e. $\dfrac{Z^{3^m}}{A^{3^m}}  + \dfrac{Z}{A}=0$. Replacing $Z$ by $A\beta^3$ in the last expression, we have $(\beta^{3^m}+\beta)^3=0$, which can be ignored as $\Trnn(\beta) \neq 0$ in the sum $S_1$. Hence, we have $S_1=0$.

 \noindent
 \textbf{Case 2.} Let $c \in \F_{3^n}\setminus\F_{3^m}$. Now if  $\delta \in \Gamma_0$, then $F(X)=X^{3^m}-\delta^{3^m}$ is clearly P$c$N, as the differential $F(X+a)-cF(X)$ is a permutation over $\F_{3^n}$. Let us now assume $\delta \not \in \Gamma_0$. Then one can easily see that $T_1$ will remain unchanged, while the expression for $T_0$ will change as follows:
 \allowdisplaybreaks
\begin{align*} 
T_{0} & =  \Tr(\beta(1-c)F(X)) \\
 & = \Tr\left(\beta(1-c)(X^{3^m}+(\delta^{2\cdot{3^{m-1}}}+\delta^{3^{m-1}+3^{2m-1}})(X^{3^{m-1}}-X^{3^{2m-1}})\right.\\
 &\qquad\qquad \left.+ (\delta^{3^{m-1}}+\delta^{3^{2m-1}})(X^{2\cdot3^{m-1}}+X^{2\cdot3^{2m-1}} +X^{3^{m-1}+3^{2m-1}})+\delta^{3^{2m-1}+2\cdot3^{m-1}})\right)\\
 & = \Tr\left(\beta(1-c)\delta^{3^{2m-1}+2\cdot3^{m-1}}\right)+ \Tr\left(\beta^3(1-c)^3(X^{3^{m+1}}+(\delta^{2\cdot{3^m}}+\delta^{3^m+1})(X^{3^m}-X)\right.\\
 &\qquad\qquad\quad\quad\quad \left.+ (\delta^{3^{m}}+\delta)(X^{2\cdot3^{m}}+X^2 +X^{3^m+1}))\right)\\
 & = \Tr\left(\beta(1-c)\delta^{3^{2m-1}+2\cdot3^{m-1}}\right)+\Tr((\beta(1-c))^{3^m} X + (\delta^2+\delta^{3^m+1})(\beta(1-c))^{3^{m+1}} X)\\
 & \quad \Tr\left(-(\delta^{2\cdot{3^m}+\delta^{3^m+1}})(\beta(1-c))^3X+(\delta^{3^{m}}+\delta)(\Trnn(\beta(1-c))^3)(X^2-X^{3^m+1})\right).
\end{align*}
 
 Now Equation~\eqref{eq35} reduces to
\begin{equation*}
 _c\Delta_F(a,b) = \frac{1}{3^n} \sum_{\beta \in \F_{3^n}} \omega^{\displaystyle{\Tr(\beta(F(a)-b- c\delta^{3^{2m-1}+2\cdot3^{m-1}})}}\sum_{X \in \F_{3^n}} \omega^{\displaystyle{\Tr(u(X^2-X^{3^m+1})+vX)}},
\end{equation*}
where $u = (\delta^{3^{m}}+\delta)(\Trnn(\beta(1-c))^3) $ and $v= (\delta^{3^m}+\delta)(a^{3^m}-a)\Trnn(\beta^3)+((1-c)\beta)^{3^m}+(\delta^2+\delta^{3^m+1})((1-c) \beta)^{3^{m+1}}-(\delta^{2\cdot3^m}+\delta^{3^m+1})(\beta(1-c))^3.$ 

Depending upon $\Trnn(\beta(1-c))=0$ or $\Trnn(\beta(1-c)) \neq 0$, we define the sums $S_0$ and $S_1$ as follows.
\allowdisplaybreaks
\begin{align*}
 S_0 & = \sum_{\substack{\beta \in \F_{3^n} \\ \Trnn(\beta(1-c))=0}} \omega^{\displaystyle{\Tr(\beta(F(a)-b- c\delta^{3^{2m-1}+2\cdot3^{m-1}}))}}\\
 & \sum_{X \in \F_{3^n}} \omega^{\displaystyle{\Tr\left(((\delta^{3^m}+\delta)(a^{3^m}-a)\Trnn(\beta^3)+((1-c)\beta)^{3^m})X\right)}}\\
 & \qquad \qquad \omega^{\displaystyle{\Tr\left(((\delta^2+\delta^{3^m+1})((1-c) \beta)^{3^{m+1}}-(\delta^{2\cdot3^m}+\delta^{3^m+1})(\beta(1-c))^3)X\right)}} \\
 & = \sum_{\substack{\beta \in \F_{3^n} \\ \Trnn(\beta(1-c))=0}} \omega^{\displaystyle{\Tr(\beta(F(a)-b- c\delta^{3^{2m-1}+2\cdot3^{m-1}}))}}\sum_{X \in \F_{3^n}} \omega^{\displaystyle{\Tr\left(-(1-c)\beta X\right)}}\\
 & \qquad\qquad \omega^{\displaystyle{\Tr\left(((\delta^{3^m}+\delta)(a^{3^m}-a)\Trnn(\beta^3)-(\delta^2+\delta^{2\cdot3^m}-\delta^{3^m+1})((1-c) \beta)^{3})X\right)}}.
\end{align*}

For computing, $S_0$, we need to know solutions of 
\[
(\delta^{3^m}+\delta)(a^{3^m}-a)\Trnn(\beta^3)-((1-c)\beta)-(\delta^2+\delta^{2\cdot3^m}-\delta^{3^m+1})((1-c) \beta)^3=0,
\]
or equivalently,
$$
(\delta^{3^m}+\delta)(a^{3^m}-a)\left(\beta-\frac{\beta(1-c)}{(1-c)^{3^m}}\right)^3-((1-c)\beta)-(\delta^2+\delta^{2\cdot3^m}-\delta^{3^m+1})((1-c) \beta)^3=0,
$$
which is same as $A\beta +B\beta^3=0$, where 
$$
A= -(1-c)
$$
and
$$
B = \frac{(\delta^{3^m}+\delta)(a^{3^m}-a)((1-c)^{3^m}-(1-c))^3-(\delta^{3^m+1}-\delta^2-\delta^{2\cdot3^m})(1-c)^{3^{m+1}+3}}{(1-c)^{3^{m+1}}}.
$$
It is easy to see that except for $\beta=0$, $A\beta +B\beta^3=0$ has a solution $\beta \in \F_{3^n}^{*}$ if $\dfrac{-A}{B}$ is a square (notice that for $a \in \F_{3^m}$, $\dfrac{-A}{B}$ is always a square). If it is a square, then $A\beta +B\beta^3=0$ has three solution in $\F_{3^n}$, namely, $\beta=0, \beta =\beta_1$ and $\beta = -\beta_1$. Hence, $S_0$ further reduces to
\begin{align*}
& 3^n\left(1 + \omega^{\displaystyle{\Tr(\beta_1(F(a)-b- c\delta^{3^{2m-1}+2\cdot3^{m-1}}))}} +  \omega^{\displaystyle{\Tr(-\beta_1(F(a)-b- c\delta^{3^{2m-1}+2\cdot3^{m-1}}))}}\right).
\end{align*}

Clearly, for those pairs of $(a,b)\in\F_{3^n}\times\F_{3^n}$ for which $b=F(a)- c\delta^{3^{2m-1}+2\cdot3^{m-1}}$, we have $S_0=3^{n+1}$; and for the other pairs of $(a,b) \in \F_{3^n} \times \F_{3^n}$, we have $S_0=3^n(1+\omega^{\Tr(\alpha)}+\omega^{\Tr(-\alpha)})$, where $\alpha = \beta_1(F(a)-b- c\delta^{3^{2m-1}+2\cdot3^{m-1}})$. Hence, the maximum value that $S_0=3^n(1+\omega^{\Tr(\alpha)}+\omega^{\Tr(-\alpha)})$ can attain is $3^{n+1}$ as $\Tr(-\alpha)=-\Tr(\alpha)$. This yields that $S_0=3^{n+1}$.

Next, we analyze the sum.
 \begin{align*}
 S_1 & = \sum_{\substack{\beta \in \F_{3^n} \\ \Trnn(\beta(1-c))\neq0}} \omega^{\displaystyle{\Tr(\beta(F(a)-b- c\delta^{3^{2m-1}+2\cdot3^{m-1}}))}} \sum_{X \in \F_{3^n}} \omega^{\displaystyle{\Tr\left(u(X^2-X^{3^m+1})+vX\right)}}\\
 & = \sum_{\substack{\beta \in \F_{3^n} \\ \Trnn(\beta(1-c))\neq0}} \omega^{\displaystyle{\Tr(\beta(F(a)-b- c\delta^{3^{2m-1}+2\cdot3^{m-1}}))}} \mathcal{W}_G(-v),
 \end{align*}
 where $\mathcal{W}_G(-v)$ is Walsh transform of trace of function $G : X \mapsto u(X^{3^0+1}-X^{3^m+1})$ at $-v$. It is clear that for  those $\beta \in \F_{3^n}$, for which $v \neq 0$, we have $\mathcal{W}_G(-v)=0$. By following the similar arguments as in the Case 1 above, one can ignore those $\beta$'s for which $v=0$. This completes the proof. 
\end{proof}

For our next theorem, we need two lemmas, the first of which was shown in~\cite{CM04} (with our notations).
\begin{lem}
\label{lem:CM04}
 Let $f(z)=z^{p^k}-\tilde a\, z-\tilde b$ in $\F_{p^n}$, $g=\gcd(n,k)$, $\ell=n/\gcd(n,k)$ and $\Trnl$ be the relative trace from $\F_{p^n}$ to $\F_{p^\ell}$. For $0\leq i\leq \ell-1$, we define $t_i=\sum_{j=i}^{\ell-2} p^{n(j+1)}$, $\alpha_0=\tilde a,\beta_0=\tilde b$. If $\ell>1$ (here, $k=\ell=2$, so we will not be concerned with $\ell=1$), then, for $1\leq r\le \ell-1$, we set
 \[
 \alpha_r={\tilde a}^{1+p^k+\cdots+p^{kr}} \text{ and } \beta_r=\sum_{i=0}^r {\tilde a}^{s_i} {\tilde b}^{p^{ki}},
 \]
 where $s_i=\sum_{j=i}^{r-1} p^{k(j+1)}$, for $0\leq i\leq r-1$ and $s_r=0$. The trinomial $f$ has no roots in $\F_{p^n}$ if and only if $\alpha_{\ell-1}=1$ and $\beta_{\ell-1}\neq 0$. If  $\alpha_{\ell-1}\neq 1$, then it has a unique root, namely $x=\beta_{\ell-1}/(1-\alpha_{\ell-1})$, and, if $\alpha_{\ell-1}=1,\beta_{\ell-1}=0$, it has $p^\ell$ roots in $\F_{p^n}$ given by $x+\delta\tau$, where $\delta\in\F_{p^\ell}$, $\tau$ is fixed in $\F_{p^n}$ with $\tau^{p^k-1}={\tilde a}$ (that is, $\tau$ is a $(p^k-1)$-root of $\tilde a$), and, for any $e\in\F^*_{p^n}$ with $\Trnl(e)\neq 0$, then
 $\displaystyle x=\frac{1}{\Trnl(e)} \sum_{i=0}^{\ell-1} \left( \sum_{j=0}^i e^{p^{kj}}\right) {\tilde a}^{t_i} {\tilde b}^{p^{ki}}$.
 \end{lem}
\begin{rem}
 We can easily simplify   the above parameters via the sum of the geometric sequence, and get
 \[
 s_i=\frac{p^{k(r+1)}-p^{k(i+1)}}{p^k-1}\ (\text{for }i<r),\quad \alpha_r={\tilde a}^{\frac{p^{k(r+1)}-1}{p^k-1}}.
 \]
 \end{rem}
 \begin{lem}
 \label{lem:AB}
 Let $c\in\F_{p^n}\setminus \F_{p^m}$, $n=2m$. With the following notations
 \begin{align*}
 A & = \left(1-c+(a^{p^m}-a)^p\left(1-\dfrac{1-c}{(1-c)^{p^m}}\right)\right)^p,\\
 B & =(a^{p^m}-a)\left(1-\dfrac{(1-c)}{(1-c)^{p^m}}\right),
\end{align*}
then there exists $a\in\F_{p^n}$ such that $A+Bd^{p-1}=0$, for some $d\in\F_{p^n}$.
 \end{lem}
 \begin{proof}
 We split our analysis in two cases depending upon the parity of $m$.
 
\noindent
{\bf Case $1$.} Let $m$ be odd.
We will show that one can find $a$ such that $A=-d^{p-1} B$, where $d= 1-(1-c)^{1-p^m} $, that is, 
\[
\left(1-c+(a^{p^m}-a)^p\left(1-\dfrac{1-c}{(1-c)^{p^m}}\right)\right)^p=
-\left(1-\dfrac{1-c}{(1-c)^{p^m}}\right)^{p-1}(a^{p^m}-a)\left(1-\dfrac{(1-c)}{(1-c)^{p^m}}\right).
\]
Since $\F_{p^n}$ is a quadratic extension over $\F_{p^m}$, then any element $\gamma\notin\F_{p^m}$ will generate $\F_{p^n}$ over $\F_{p^m}$. In particular, one can take $\gamma=1-c$. Write ${a=x\gamma+y}$, where $x,y\in\F_{p^m}$. 
First, we note that $a^{p^m}-a=x(\gamma^{p^m}-\gamma)$ and $(\gamma^{p^m}-\gamma)^{p^m}=\gamma-\gamma^{p^m}$, and thus, $\cG^2\in\F_{p^m}$, where $\cG:=\gamma^{p^m}-\gamma\notin\F_{p^m}$.
The above displayed equation then becomes 
\begin{align*}
x^{p^2} \left(1-\gamma^{p(1-p^m)} \right) (\gamma^{p^m}-\gamma)^{p^2}+x  \left(1-\gamma^{p(1-p^m)} \right)(\gamma^{p^m}-\gamma)+\gamma^p=0.
\end{align*}
Dividing by $(1-\gamma^{p(1-p^m)})$,
 and observing that $\frac{\gamma}{1-\gamma^{1-p^m}}=\frac{\gamma^{p^m+1}}{\cG}$, we get 
 \[
 x^{p^2}\cG^{p^2}+x\cG+\left(\frac{\gamma^{p^m+1}}{\cG}\right)^p.
 \]
Multiplying by $\cG^{p^2}$, throughout, we obtain
\begin{equation}
\label{eq:last_lin}
z^{p^2} + z\,\cG^{p^2-1}+\left(\gamma^{p^{m}+1}{\cG^{p-1}}\right)^p=0,
\end{equation}
where $z=x\cG^2\in\F_{p^m}$, since both $x,\cG^2\in\F_{p^m}$, as well as $\gamma^{p^m+1} {\cG^{p-1}}\in\F_{p^m}$, since both $\gamma^{p^m+1}{\cG^{p-1}}\in\F_{p^m}$. Hence, this last equation is over $\F_{p^m}$.
  
We now use Lemma~\ref{lem:CM04}, so we let $k=2,g=1$ (here, we use that $m$ is odd), $\ell=m$, $\tilde a=-\cG^{p^2-1}$, $\tilde b=-\gamma^{p^m+1}{\cG^{p-1}}$.  To show the existence of solutions $z$ for Equation~\eqref{eq:last_lin} (and hence our initial claim) it will be sufficient to show that  
$\alpha_{m-1}=\displaystyle {\tilde a}^{\frac{p^m-1}{p^2-1}}\neq 1$, that is, 
\[
 \left(-\cG^{p^2-1}\right)^{\frac{p^m-1}{p^2-1}}=\cG^{p^m-1}\neq 1,
\]
which is obviously true since $\cG\notin\F_{p^m}$. 
 
 \noindent
{\bf Case $2$.} Let $m$ be even. We will now show that $d$ can be taken to be $(a^{p^m}-a)\left(1-\dfrac{1-c}{(1-c)^{p^m}}\right)$, that is,
  \[
\left(1-c+(a^{p^m}-a)^p\left(1-\dfrac{1-c}{(1-c)^{p^m}}\right)\right)^p=
-\left(1-\dfrac{1-c}{(1-c)^{p^m}}\right)^{p}(a^{p^m}-a)^p,
\]
which is equivalent (by taking $p$-roots) to
\[
 1-c+(a^{p^m}-a)^p\left(1-\dfrac{1-c}{(1-c)^{p^m}}\right) =
-\left(1-\dfrac{1-c}{(1-c)^{p^m}}\right)(a^{p^m}-a).
\]
 As before, take $\gamma=1-c$ of order 2 over $\F_{p^m}$ and write   $a=x+y\gamma$, where $x,y\in\F_{p^m}$, so 
  $a^{p^m}-a=x(\gamma^{p^m}-\gamma)$ and  $\cG^2\in\F_{p^m}$, where $\cG:=\gamma^{p^m}-\gamma\notin\F_{p^m}$. Our equation then becomes
  \[
  x^{p} \left(1-\gamma^{1-p^m} \right) (\gamma^{p^m}-\gamma)^{p}+x  \left(1-\gamma^{1-p^m} \right)(\gamma^{p^m}-\gamma)+\gamma=0.
  \]
  Multiplying by $\cG^p/(1-\gamma^{1-p^m} )$ and using the substitution $z=x\cG^2$, we obtain an equation in~$z$, where all the coefficients are in $\F_{p^m}$, namely
\begin{equation}
\label{eq:last_lin2}
  z^p+\cG^{p-1}\, z+\gamma^{p^m+1}\cG^{p-1}=0.
  \end{equation}
  We next use Lemma~\ref{lem:CM04},  with $k=1,g=1$, $\ell=m$, $\tilde a=-\cG^{p-1}$, $\tilde b=-\gamma^{p^m+1}\cG^{p-1}$. We need that the parameter $\alpha_{\ell-1}=\alpha_{m-1}\neq 1$, that is
  \[
 \alpha_{m-1}= {\tilde a}^{\frac{p^m-1}{p-1}}= 
 \left(\cG^{p-1} \right)^{\frac{p^m-1}{p-1}}=\cG^{p^m-1} \neq 1,
 \]
 which surely is true since $\cG\notin\F_{p^m}$.
 Our claim is shown.
 \end{proof}

In the following theorem, we discuss another class of permutation polynomials given in~Lemma~\ref{L03} over finite fields of odd characteristic~$p$.
\begin{thm} \label{T12}
  Let $F(X)=(X^{p^m}-X+\delta)^{p^{m+1}+1}+X$ over $\F_{p^{n}}$, where $n=2m$ and $\delta\in\F_{p^n}$, where $\Trnn(\delta)=0$ or $\frac{\Trnn(\delta)-1}{\Trnn(\delta)}$ is a $(p-1)$-th power in $\F_{p^m}$. Then
  \begin{enumerate}
   \item[$(1)$] $F$ is  P$c$N for all $c \in \F_{p^m} \setminus \{1\}$;
   \item[$(2)$]  F is of $c$-differential uniformity $p$ for all $c \in \F_{p^n}\setminus\F_{p^m}$.    
 \end{enumerate}
  \end{thm}
 \begin{proof}
We can write $F$ as given in the following expression,
\allowdisplaybreaks
\begin{align*}
 F(X)= & X^{p^m+p}+X^{p^{m+1}+1}-X^{p+1}-X^{p^m+p^{m+1}}+\delta X^p-\delta X^{p^{m+1}} +\delta^{p^{m+1}}X^{p^m}\\
 & \qquad \quad + (1-\delta^{p^{m+1}})X + \delta^{p^{m+1}+1},\\
  = & \Trnn(X^{p^{m+1}+1}-X^{p+1})+\delta X^p-\delta X^{p^{m+1}} +\delta^{p^{m+1}}X^{p^m}+ (1-\delta^{p^{m+1}})X + \delta^{p^{m+1}+1}.
\end{align*}

We know that for any $(a, b) \in \F_{p^{n}} \times \F_{p^{n}}$,
the $c$-DDT entry $_c\Delta_F(a, b)$ is given by the number of solutions $X \in \F_{p^{n}}$ of the following equation
\begin{equation}\label{eq41}
 F(X+a)-cF(X)=b,
\end{equation}
or, equivalently,
\begin{equation*}
 (1-c)F(X)+ \Trnn(aX^{p^{m+1}}+(a^{p^{m+1}}-a^p)X-aX^p) + F(a)- \delta^{p^{m+1}+1} = b.
 \end{equation*}
From Equation~\eqref{ddtw}, one can see that the number of solutions $X \in \F_{p^n}$ of the above Equation~\eqref{eq41}, $_c\Delta_F(a,b)$, is given by 
\allowdisplaybreaks
\begin{align*}
 & \displaystyle{\frac{1}{p^n} \sum_{\beta \in \F_{p^n}} \sum_{X \in \F_{p^n}} \omega^{\displaystyle{\Tr\left(\beta((1-c)F(X)+ \Trnn(aX^{p^{m+1}}+(a^{p^{m+1}}-a^p)X-aX^p))\right)}}} \\
 & \hspace{2.5cm} \omega^{\displaystyle{\Tr(\beta(F(a)-\delta^{p^{m+1}+1}-b))}},
\end{align*}
where $\omega =e^{2\pi i/p}$. Equivalently,
\allowdisplaybreaks
\begin{align*}
  _c\Delta_F(a,b) & = \dfrac{1}{p^n} \sum_{\beta \in \F_{p^n}} \omega^{\displaystyle{\Tr\left(\beta\left(F(a)-b-\delta^{p^{m+1}+1}\right)\right)}}\\
& \sum_{X \in \F_{p^n}} \omega^{\displaystyle{\Tr\left(\beta((1-c)F(X)+ \Trnn(aX^{p^{m+1}}+(a^{p^{m+1}}-a^p)X-aX^p)\right)}},\\
& = \frac{1}{p^n} \sum_{\beta \in \F_{p^n}} \omega^{\displaystyle{\Tr\left(\beta\left(F(a)-b-\delta^{p^{m+1}+1}\right)\right)}}\sum_{X \in \F_{p^n}} \omega^{\displaystyle{T_0+T_1}},
 \end{align*}
where $T_0 = \Tr(\beta(1-c)F(X))$ and $T_1 = \Tr(\beta\Trnn(aX^{p^{m+1}}+(a^{p^{m+1}}-a^p)X-aX^p))$. 

\noindent
\textbf{Case 1.} Let $c\in \F_{p^m}\setminus \{1\}$. We first compute $T_0$ and $T_1$ as follows
\allowdisplaybreaks
\begin{equation*} 
\begin{split}
T_{1} & = \Tr(\beta\Trnn(aX^{p^{m+1}}+(a^{p^{m+1}}-a^p)X-aX^p)) \\
 & = \Tr\left(\beta(aX^{p^{m+1}}+a^{p^m}X^{p}+(a^{p^{m+1}}-a^p)X+(a^p-a^{p^{m+1}})X^{p^m}-aX^p-a^{p^m}X^{p^{m+1}})\right) \\
 & = \Tr\left((\beta(a-a^{p^m}))^{p^{m-1}}X+(\beta(a^{p^m}-a))^{p^{2m-1}}X\right) \\
 & \qquad \qquad \qquad + \Tr\left(\beta(a^{p^{m+1}}-a^p)X+(\beta(a^p-a^{p^{m+1}}))^{p^m}X\right)\\
 & = \Tr\left(((a-a^{p^m})^{p^{m-1}}\Trnn(\beta^{p^{m-1}}) +(a^{p^{m+1}}-a^p)\Trnn(\beta))X\right),
\end{split}
\end{equation*}
and
\allowdisplaybreaks
\begin{align*} 
T_{0} & =  \Tr(\beta(1-c)F(X)) \\
 & = \Tr\left(\beta(1-c)(\Trnn(X^{p^{m+1}+1}-X^{p+1})+\delta X^p-\delta X^{p^{m+1}} +\delta^{p^{m+1}}X^{p^m}+\right.\\
 &\qquad\qquad \left. (1-\delta^{p^{m+1}})X + \delta^{p^{m+1}+1})\right)\\
 & = \Tr\left((\beta(1-c))^{p^{m-1}} \Trnn(X^{p^{m-1}+1})-\beta(1-c)\Trnn(X^{p+1})+\right.\\
 &\qquad\qquad \left. \beta(1-c)(\delta X^p-\delta X^{p^{m+1}} +\delta^{p^{m+1}}X^{p^m}+(1-\delta^{p^{m+1}})X + \delta^{p^{m+1}+1})\right)\\
 & = \Tr\left(\beta(1-c)\delta^{p^{m+1}+1}\right)+ \Tr\left(\Trnn(\beta(1-c))^{p^{m-1}}X^{p^{m-1}+1}-\Trnn(\beta(1-c))X^{p+1}\right.\\
 & \left. + (\beta(1-c)\delta)^{p^{2m-1}}X-(\beta(1-c)\delta)^{p^{m-1}}X + (\beta(1-c)\delta^{p^{m+1}})^{p^m}X+\beta(1-c)(1-\delta^{p^{m+1}})X\right).
\end{align*}
 Let us assume that
 \allowdisplaybreaks
 \begin{align*}
  u_1 & = \Trnn(\beta(1-c))^{p^{m-1}}=((1-c)\Trnn(\beta))^{p^{m-1}}\\
  u_2 & = \Trnn(\beta(1-c)) =(1-c)\Trnn(\beta)\\
  v & = (1-c)^{p^{m-1}}(\beta\delta)^{p^{2m-1}}-(\beta(1-c)\delta)^{p^{m-1}}+ (1-c)(\beta\delta^{p^{m+1}})^{p^m}+\beta(1-c)(1-\delta^{p^{m+1}})\\
  & \qquad +(a-a^{p^m})^{p^{m-1}}\Trnn(\beta^{p^{m-1}}) +(a^{p^{m+1}}-a^p)\Trnn(\beta).
 \end{align*}
Hence, we have
\begin{equation*}
 _c\Delta_F(a,b) = \frac{1}{p^n} \sum_{\beta \in \F_{p^n}} \omega^{\displaystyle{\Tr(\beta(F(a)-b-c\delta^{p^{m+1}+1}))}}\sum_{X \in \F_{p^n}} \omega^{\displaystyle{\Tr(u_1 X^{p^{m-1}+1}-u_2X^{p+1}+vX)}}.
\end{equation*}

Further, splitting the above sum depending on whether $\Trnn(\beta)$ is~$0$ or not, we get
\allowdisplaybreaks
\begin{align*}
 p^n\,_c\Delta_F(a,b) & =  \sum_{\substack{\beta \in \F_{p^n} \\ \Trnn(\beta)=0}} \omega^{\displaystyle{\Tr(\beta(F(a)-b-c\delta^{p^{m+1}+1}))}}\\
 &  \sum_{X \in \F_{p^n}} \omega^{\displaystyle{\Tr\left(((\beta(1-c)\delta)^{p^{2m-1}}-(\beta(1-c)\delta)^{p^{m-1}})X\right)}} \\
 & \omega^{\displaystyle{\Tr\left( ((1-c)(\beta\delta^{p^{m+1}})^{p^m}+\beta(1-c)(1-\delta^{p^{m+1}}))X\right)}} \\
 & +  \sum_{\substack{\beta \in \F_{p^n} \\ \Trnn(\beta)\neq 0}} \omega^{\displaystyle{\Tr(\beta(F(a)-b-c\delta^{p^{m+1}+1}))}}\sum_{X \in \F_{p^n}} \omega^{\displaystyle{\Tr(u_1X^{p^{m-1}+1}-u_2X^{p+1}+vX)}}\\
 &=  S_0 + S_1,
\end{align*}
where $S_0,S_1$ are the two inner sums, as given below. First, we write $S_0$ as
\allowdisplaybreaks
\begin{align*}
 S_0 & = \sum_{\substack{\beta \in \F_{p^n} \\ \Trnn(\beta)=0}} \omega^{\displaystyle{\Tr(\beta(F(a)-b-c\delta^{p^{m+1}+1}))}}\\
 & \sum_{X \in \F_{p^n}} \omega^{\displaystyle{\Tr\left((\beta(1-c)\delta)^{p^{2m-1}}-(\beta(1-c)\delta)^{p^{m-1}}+ (1-c)(\beta(1-\delta^{p^{m+1}})+\beta^{p^m}\delta^p))X\right)}} \\
 & = \sum_{\substack{\beta \in \F_{p^n} \\ \Trnn(\beta)=0}} \omega^{\displaystyle{\Tr(\beta(F(a)-b-c\delta^{p^{m+1}+1}))}}\\
 & \sum_{X \in \F_{p^n}} \omega^{\displaystyle{\Tr\left(-(\beta(1-c))^{p^{m-1}}\Trnn(\delta^{p^{m-1}})X+ (1-c)\beta(1-\Trnn(\delta^p))X\right)}} \\
&  = p^n.
\end{align*}
We give the justification for the above equality for both possibilities of $\delta$, i.e., $\Trnn(\delta)=0$ and $\dfrac{\Trnn(\delta)-1}{\Trnn(\delta)}$ is a $(p-1)$-th power in $\F_{p^m}$ as follows. 
First consider the case when $\Trnn(\delta)=0$, then the inner sum in $S_0$ becomes zero, except for $\beta=0$ and hence we have $S_0=p^n$.   

Next, we consider the case when $\Trnn(\delta)\neq0$ and $ \gamma^{p-1}=\dfrac{\Trnn(\delta)-1}{\Trnn(\delta)}$ for some $\gamma \in \F_{p^m}$. Then to show that $S_0=p^n$, it is sufficient to show that the below equation
\begin{equation}\label{eq42}
 -(\beta(1-c))^{p^{m-1}}\Trnn(\delta^{p^{m-1}})+ (1-c)\beta(1-\Trnn(\delta^p))=0
\end{equation}
has only one solution $\beta \in \F_{p^n}$ satisfying $\Trnn(\beta)=0$. Now raising the Equation~\eqref{eq42} to the power $p$, we get 
$$-(\beta(1-c))^{p^m}\Trnn(\delta^{p^m})+ ((1-c)\beta)^p(1-\Trnn(\delta))^{p^2}=0,$$
or equivalently,
$$
\beta(1-c)^{p^m}\Trnn(\delta^{p^m})+ ((1-c)\beta)^p(1-\Trnn(\delta))^{p^2}=0,
$$
which gives us
$$
\beta^{p-1}=\dfrac{- \Trnn(\delta)}{(1-c)^{p-1}(1-\Trnn(\delta))^{p^2}}=\dfrac{1}{\left(\gamma(1-c)(1-\Trnn(\delta))^{p+1}\right)^{p-1}}.
$$
Clearly, from the above equation we have $\beta = \dfrac{\alpha}{\gamma(1-c)(1-\Trnn(\delta))^{p+1}}$ for $\alpha \in \F_{p^n}$ such that $\alpha^{p-1}=1$. One can easily see that
$$
\Trnn(\beta)= \dfrac{\alpha+\alpha^{p^m}}{\gamma(1-c)(1-\Trnn(\delta))^{p+1}} \neq 0.
$$
Hence, the only possible solution of Equation~\eqref{eq42} is $\beta=0$, making $S_0=p^n$.

Next, we consider the sum $S_1$.
\allowdisplaybreaks
\begin{align*}
 S_1 & = \sum_{\substack{\beta \in \F_{p^n} \\ \Trnn(\beta)\neq 0}} \omega^{\displaystyle{\Tr(\beta(F(a)-b-c\delta^{p^{m+1}+1}))}} \sum_{X \in \F_{p^n}} \omega^{\displaystyle{\Tr\left(u_1X^{p^{m-1}+1}-u_2X^{p+1}+vX\right)}}\\
 & = \sum_{\substack{\beta \in \F_{p^n} \\ \Trnn(\beta)\neq0}} \omega^{\displaystyle{\Tr(\beta(F(a)-b-c\delta^{p^{m+1}+1}))}} \mathcal{W}_G(-v),
 \end{align*}
 where $\mathcal{W}_G(-v)$ is Walsh transform of trace of function $G : X \mapsto u_1X^{p^{m-1}+1}-u_2X^{p+1}$ at $-v$. Now, from Lemma~\ref{walsh}, the absolute square of Walsh transform coefficient of $G$ is given by
 \begin{equation*} \lvert \mathcal{W}_G(-v) \rvert^2 =
  \begin{cases}
   p^{n+\ell} &~\mbox{if}~G(X)+\Tr(vX)\equiv0~\mbox{on Ker}~(L)  \\
    0 &~\mbox{otherwise},
  \end{cases}
 \end{equation*}
 where $\ell$ is dimension of kernel of the linearized polynomial 
 $$
 L(X)=u_2(X^{p^m}-X)^p-u_1(X^{p^m}-X)^{p^{m-1}}.
 $$ 
It is easy to see that $\F_{p^m} \subseteq$ Ker$(L) $. Thus, if we can show that $G(X)+\Tr(vX) \neq 0$ for all $X \in \F_{p^m}$, then $S_1=0$. We shall now make efforts to prove that $G(X)+\Tr(vX)$ is not identically zero on $\F_{p^m}$. For $X \in \F_{p^m}$, the polynomial $G(X)+\Tr(vX)$ gets reduced to a polynomial over $\F_{p^m}$ given  by
 $$
 G(X)+\Tr(vX)=  u_1X^{p^{m-1}+1}-u_2X^{p+1} + (v+v^{p^m})X+((v+v^{p^m})X)^p+ \cdots +((v+v^{p^m})X)^{p^{m-1}}.
 $$
If $m=1$, then the above equation becomes  a quadratic equation $G(X)+\Tr(vX)=  (u_1-u_2)X^{2} + (v+v^{p})X$ over $\F_{p}$, which does not completely vanish on $\F_{p}$. For $m \geq 2$, it is easy to observe that the degree of the polynomial $G(X)+\Tr(vX)$ is strictly less than $p^m-1$. As a consequence, $G(X)+\Tr(vX) \neq 0$ for all $X \in \F_{p^m}$.

\noindent
\textbf{Case 2.} Let $c\in \F_{p^n}\setminus \F_{p^m}$. Then we have 
 \begin{equation}
 \label{eq44}
 _c\Delta_F(a,b) = \frac{1}{p^n} \sum_{\beta \in \F_{p^n}} \omega^{\displaystyle{\Tr(\beta(F(a)-b-c\delta^{p^{m+1}+1}))}}\sum_{X \in \F_{p^n}} \omega^{\displaystyle{\Tr(u_1 X^{p^{m-1}+1}-u_2X^{p+1}+vX)}},
\end{equation}
where,
\begin{align*}
  u_1 & = \Trnn(\beta(1-c))^{p^{m-1}}\\
  u_2 & = \Trnn(\beta(1-c))\\
  v & = (\beta(1-c)\delta)^{p^{2m-1}}-(\beta(1-c)\delta)^{p^{m-1}}+ (\beta(1-c)\delta^{p^{m+1}})^{p^m}+\beta(1-c)(1-\delta^{p^{m+1}})\\
  & \qquad +(a-a^{p^m})^{p^{m-1}}\Trnn(\beta^{p^{m-1}}) +(a^{p^{m+1}}-a^p)\Trnn(\beta).
 \end{align*}
 
 Now we analyze the sum in Equation~\eqref{eq44}, by splitting it in two cases depending on whether $\Trnn(\beta(1-c))$ is~$0$ or not, we get
\allowdisplaybreaks
\begin{align*}
p^n\, _c\Delta_F(a,b) & =  \sum_{\substack{\beta \in \F_{p^n} \\ \Trnn(\beta(1-c))=0}} \omega^{\displaystyle{\Tr(\beta(F(a)-b-c\delta^{p^{m+1}+1}))}} \\
 & \quad \sum_{X \in \F_{p^n}} \omega^{\displaystyle{\Tr\left(((\beta(1-c)\Trnn(\delta))^{p^{2m-1}}+ \beta(1-c)(1-\Trnn(\delta^p))X)\right)}} \\
 & \qquad\qquad \omega^{\displaystyle{\Tr\left(((a-a^{p^m})^{p^{m-1}}\Trnn(\beta^{p^{m-1}}) +(a^{p^{m+1}}-a^p)\Trnn(\beta))X\right)}}  \\
 &\qquad   +  \sum_{\substack{\beta \in \F_{p^n} \\ \Trnn(\beta(1-c))\neq 0}} \omega^{\displaystyle{\Tr(\beta(F(a)-b-c\delta^{p^{m+1}+1}))}}   \\
 &   \qquad \qquad \qquad \sum_{X \in \F_{p^n}} \omega^{\displaystyle{\Tr(u_1X^{p^{m-1}+1}-u_2X^{p+1}+vX)}} \\
 &=  S_0 + S_1,
\end{align*}
where $S_0$ and $S_1$ are the two inner sums, corresponding to $\Trnn(\beta(1-c))=0$ and  $\Trnn(\beta(1-c))\neq 0$ not respectively. First, we take $\Trnn(\delta)=0$ and compute $S_0$,
\allowdisplaybreaks
\begin{align*}
 S_0 & = \sum_{\substack{\beta \in \F_{p^n} \\ \Trnn(\beta(1-c))=0}} \omega^{\displaystyle{\Tr(\beta(F(a)-b-c\delta^{p^{m+1}+1}))}}\\
 & \sum_{X \in \F_{p^n}} \omega^{\displaystyle{\Tr\left((\beta(1-c)+(a-a^{p^m})^{p^{m-1}}\Trnn(\beta^{p^{m-1}}) +(a^{p^{m+1}}-a^p)\Trnn(\beta))X\right)}}.
\end{align*}
Now to compute the inner sum in $S_0$, we need to look for solutions $\beta \in \F_{p^n}$ of the equation given below,
 \begin{equation}\label{eq45}
  \beta(1-c)+(a-a^{p^m})^{p^{m-1}}\Trnn(\beta^{p^{m-1}}) +(a^{p^{m+1}}-a^p)\Trnn(\beta)=0.
 \end{equation}
Further substituting $\beta^{p^m} = \dfrac{-(1-c)}{(1-c)^{p^m}}\beta$, we reduced the above equation as
 \[
 \left(1-c+(a^{p^m}-a)^p\left(1-\dfrac{1-c}{(1-c)^{p^m}}\right)\right)\beta+\left((a-a^{p^m})\left(1-\dfrac{(1-c)}{(1-c)^{p^m}}\right)\right)^{p^{m-1}}\beta^{p^{m-1}}=0, 
 \]
and by further raising the above equation to the power $p$, we can write it as
\begin{align*}
 & \left(1-c+(a^{p^m}-a)^p\left(1-\dfrac{1-c}{(1-c)^{p^m}}\right)\right)^p\beta^p \\
 & + \left((a-a^{p^m})\left(1-\dfrac{(1-c)}{(1-c)^{p^m}}\right)\right)^{p^m}\left(\dfrac{-(1-c)}{(1-c)^{p^m}}\right)\beta=0.
\end{align*}
For simplicity, we write the above equation as $A\beta^p+ B\beta=0$, where
\begin{align*}
 A & = \left(1-c+(a^{p^m}-a)^p\left(1-\dfrac{1-c}{(1-c)^{p^m}}\right)\right)^p,\\
 B & =(a^{p^m}-a)\left(1-\dfrac{(1-c)}{(1-c)^{p^m}}\right).
\end{align*}
Notice that the equation $A\beta^p+ B\beta=0$, except for $\beta=0$, has $p-1$ solutions in $\F_{p^n}$ only if $\dfrac{-B}{A}$ is $(p-1)$th power of some element in $\F_{p^n}$, which is true via   Lemma~\ref{lem:AB}.
  
Thus, we have $\beta = \eta d$, where $\eta \in \F_{p^n}$ satisfying $\eta^{p-1}=1$. Hence, $A\beta^p+ B\beta=0$ has $p$ solutions in $\F_{p^n}$ namely, $\beta=0, \beta=\beta_1,\beta=\beta_2, \ldots,\beta=\beta_{p-1}$. With this, we now have
\begin{align*}
 S_0 = p^n\left(1 + \omega^{\Tr(\beta_1(F(a)-b-c\delta^{p^{m+1}+1}))}+ \cdots +\omega^{\Tr(\beta_{p-1}(F(a)-b-c\delta^{p^{m+1}+1}))}\right).
\end{align*}
It is easy to see that for those pairs $(a,b)\in\F_{p^n}\times\F_{p^n}$ for which $b=F(a)-c\delta^{p^{m+1}+1}$, we have $S_0=p^{n+1}$; and for the other pairs  $(a,b) \in \F_{p^n} \times \F_{p^n}$, we have 
$$S_0=p^n\left(1+\omega^{\Tr(\alpha_1)}+\omega^{\Tr((p-1)\alpha_1)}+\cdots +\omega^{\Tr\left(\alpha_{\frac{p-1}{2}}\right)}+\omega^{\Tr\left((p-1)\alpha_{\frac{p-1}{2}}\right)}\right),$$ 
where $\alpha_i = \beta_i(F(a)-b-c\delta^{p^{m+1}+1})$ for $i=1,2,\ldots,\dfrac{p-1}{2}$. Observe that if $\alpha_i \neq \alpha_j$ and $\Tr(\alpha_i)=\Tr(\alpha_j)$, then $\Tr(d(F(a)-b-c\delta^{p^{m+1}+1}))=0$, and hence $\Tr(\eta d(F(a)-b-c\delta^{p^{m+1}+1}))=0$ for all $\eta \in \F_{p^n}$ with $\eta^{p-1}=1$, making $S_0 = p^{n+1}$. Thus, we assume $\Tr(\alpha_i) \neq \Tr(\alpha_j)$ for $i \neq j$. This yields that $S_0=p^n(1+\omega+\omega^2+\omega^3+\cdots+\omega^{p-1})=0.$

We next consider $S_0$ when $\dfrac{\Trnn(\delta)-1}{\Trnn(\delta)}=\gamma^{p-1}$ for $\gamma \in \F_{p^m}$. Observe that $\gamma \neq 1$. Now, $S_0$ is given by
\begin{align*}
 S_0 & = \sum_{\substack{\beta \in \F_{p^n} \\ \Trnn(\beta(1-c))=0}} \omega^{\displaystyle{\Tr(\beta(F(a)-b-c\delta^{p^{m+1}+1}))}}\\
 & \sum_{X \in \F_{p^n}} \omega^{\displaystyle{\Tr\left(\left(\left(\frac{\beta(1-c)}{1-\gamma^{p-1}}\right)^{p^{2m-1}}+ \beta(1-c)\left(1-\frac{1}{1-\gamma^{p-1}}\right)^p\right)X\right)}}\\
 & \hspace{1.15cm}  \omega^{\displaystyle{\Tr\left(\left((a-a^{p^m})^{p^{m-1}}\Trnn(\beta^{p^{m-1}}) -(a-a^{p^m})^p\Trnn(\beta)\right)X\right)}}.
\end{align*}
Using a similar technique as in the case of $\Trnn(\delta)=0$ above, we compute $S_0$ by considering the solutions $\beta \in \F_{p^n}$ of the following equation:
\begin{align*}
 & \left(\frac{\beta(1-c)}{1-\gamma^{p-1}}\right)^{p^{2m-1}}+ \beta(1-c)\left(1-\frac{1}{1-\gamma^{p-1}}\right)^p \\
 & \qquad \qquad \qquad\qquad  +(a-a^{p^m})^{p^{m-1}}\Trnn(\beta^{p^{m-1}}) -(a-a^{p^m})^p\Trnn(\beta)=0. 
\end{align*}
After substituting $\beta^{p^m} = \dfrac{-(1-c)}{(1-c)^{p^m}}\beta = \tilde c \beta$, where $\tilde c = \dfrac{-(1-c)}{(1-c)^{p^m}}$ and raising it to the power $p$, we obtain the following equation,
\begin{align*}
 & \left((1-c)\left(1-\frac{1}{1-\gamma^{p-1}}\right)^p-(a-a^{p^m})^p(1+\tilde c)\right)^p \beta^p\\
 & \qquad \qquad \qquad+ \left({ \dfrac{1-c}{1-\gamma^{p-1}}}+\tilde c(a^{p^m}-a)(1+\tilde c )^{p^m} \right)\beta =0.
\end{align*}
It is easy to see that if $a \in \F_{p^m}$, then the above equation always has $p$ solutions. Otherwise, it will either have exactly one solution $\beta =0$ or $p$ solutions in $\F_{p^n}$. Then using the same arguments as in  the case of $\Trnn(\delta)=0$, we get that either $S_0=p^n$, or $S_0= p^{n+1}$. Next, we have
\allowdisplaybreaks
\begin{align*}
 S_1 & =  \sum_{\substack{\beta \in \F_{p^n} \\ \Trnn(\beta(1-c))\neq 0}} \omega^{\displaystyle{\Tr(\beta(F(a)-b-c\delta^{p^{m+1}+1}))}} \sum_{X \in \F_{p^n}} \omega^{\displaystyle{\Tr(u_1X^{p^{m-1}+1}-u_2X^{p+1}+vX)}}.
\end{align*}
By following similar arguments as given for $S_1$ in the Case 1, we can easily show that $S_1=0$. This completes the proof.
 \end{proof}

\section{Conclusion}\label{S5}

The $c$-differential uniformity as introduced in~\cite{Ellingsen} is a measure of statistical biases in the distribution of differences. 
In this paper we concentrate on a few permutations polynomials and investigate some equations over binary and odd characteristic fields, as they are connected to those permutations being PcN/APcN, or other low differential uniformity. In particular, our work adds to the very few known classes of PcN functions over binary fields. The used methods include discrete Fourier transforms, Weil sums and a very detailed analysis of those equations. We suspect the methods may be of independent interest.


\begin{thebibliography}{99}

\bibitem{AMS22}
N. Anbar, T. Kalayci, W. Meidl, C. Riera, P. St\u anic\u a, {\it P$\wp$N functions, complete mappings and quasigroup difference sets}, arXiv (2022), \url{https://arxiv.org/abs/2212.12943}.

\bibitem{BKM22}
D. Bartoli, L. K\"olsch, G. Micheli,
{\em Differential biases, $c$-differential uniformity, and their relation to differential attacks},
 arXiv (2022), \url{https://arxiv.org/pdf/2208.03884.pdf}.

\bibitem{Biham91} E. Biham, A. Shamir, {\it Differential cryptanalysis of DES-like cryptosystems,} J. Cryptol. 4:1 (1991), 3--72.

\bibitem{Borisov} N. Borisov, M. Chew, R. Johnson, D. Wagner, {\it Multiplicative Differentials}, In International Workshop on Fast Software Encryption (pp. 17-33). Springer, Berlin, Heidelberg (2002).

\bibitem{CM04}
R. S. Coulter, M. Henderson, {\em A note on the roots of trinomials over a finite field}, Bull. Austral. Math. Soc. 69 (2004), 429--432.

\bibitem{Ding_C_13}  C. Ding, T. Helleseth, {\it Optimal ternary cyclic codes from monomials,} IEEE Trans. Inf. Theory 59 (2013), 5898--5904.

\bibitem{Ding_Co_06}  C. Ding, J. Yuan, {\it A family of skew Hadamard difference sets,} J. Comb. Theory, Ser. A 113 (2006) 1526--1535.

\bibitem{Ellingsen}  P. Ellingsen, P. Felke, C. Riera, P. St\u anic\u a, A. Tkachenko, {\it $C$-differentials, multiplicative uniformity, and (almost) perfect $c$-nonlinearity,} IEEE Trans. Inf. Theory 66:6 (2020), 5781--5789.

\bibitem{HPRS20} S. U. Hasan,  M. Pal,  C. Riera,  P. St\u anic\u a, {\it On   the $c$-differential   uniformity   of   certain   maps   over   finite   fields}, Des. Codes Cryptogr. 89 (2021), 221--239.

\bibitem{HPS1} S. U. Hasan, M. Pal, P. St\u anic\u a, {\it On the $c$-differential uniformity and boomerang uniformity of two classes of permutation polynomials}, IEEE Trans. Inf. Theory 68 (2022), 679--691.

\bibitem{TH} T. Helleseth, A. Kholosha, {\it Monomial and quadratic bent functions over the finite fields of odd characteristic.} IEEE Trans. Inf. Theory 52, no. 5 (2006), 2018--2032.

\bibitem{JKK2} J. Jeong, N. Koo, S. Kwon, {\it Investigations of $c$-differential uniformity of permutations with Carlitz rank 3}, Finite Fields Appl. 86 (2023), 102145.

\bibitem{JKK} J. Jeong, N. Koo, S. Kwon, {\it On non-monomial AP$c$N permutations over finite
fields of even characteristic}, Finite Fields Appl. 89 (2023), 102196.


\bibitem{LMW} J. Lahtonen, G. McGuire, H.N. Ward, {\it Gold and Kasami-Welch functions, quadratic forms and bent functions,} Adv. Math. Commun. 1:2 (2007), 243--250.

\bibitem{Chapuy_C_07} Y. Laigle-Chapuy, {\it Permutation polynomials and applications to coding theory,} Finite Fields Appl. 13 (2007), 58--70.

\bibitem{Liu} Q. Liu, Z. Huang, J. Xie, X. Liu, J. Zou, {\it The $c$-differential uniformity and boomerang uniformity of three classes of permutation polynomials over $\mathbb F_{2^n}$}, Finite Fields Appl. 89 (2023), 102212.

\bibitem{CPC} C. Li, C. Riera, P. St\u anic\u a, {\it Low $c$-differentially uniform functions via an extension of Dillon's switching method}, arXiv (2022), {\url{https://arxiv.org/abs/2204.08760}}; Extended Abstract, Boolean Functions \& Applic. (BFA'22), 2022, Paper \#1.

\bibitem{LWC}  L. Li, S. Wang, C. Li, X. Zeng, {\it Permutation polynomials $(x^{p^m}-x+\delta)^{s_1} + (x^{p^m}-x+\delta)^{s_2}+x$ over $\F_{p^n}$,} Finite Fields  Appl.  51  (2018), 31--61.

\bibitem{Lidl_Cr_84}  R. Lidl, W.B. Mullen, {\it Permutation polynomials in RSA-cryptosystems,} in: Advances in Cryptology, Plenum, New York, 1984, pp. 293--301.

\bibitem{MRS} S. Mesnager, C. Riera, P. St\u anic\u a, H. Yan, and Z. Zhou, {\it Investigations on $c$-(Almost) Perfect Nonlinear Functions}, IEEE Trans. Inf. Theory  67:10 (2021),  6916--6925.

\bibitem{Nyberg} K. Nyberg, {\it Differentially uniform mappings for cryptography,} In T. Helleseth (ed), Advances in Cryptology-EUROCRYPT’93, LNCS 765, pp. 55--64, Springer, Heidelberg (1994).


\bibitem{Schwenk_Cr_98} J. Schwenk, K. Huber, {\it Public key encryption and digital signatures based on permutation polynomials,} Electron. Lett. 34 (1998), 759--760.

\bibitem{PSweil} P. St\u anic\u a, {\it Using double Weil sums in finding the $c$-boomerang connectivity table for monomial functions on finite fields,}  Appl. Algebra Eng. Commun. Comput. 34, 581--602 (2023).

\bibitem{StanicaPG2021} P. St\u anic\u a, C. Riera, A. Tkachenko, {\it Characters, Weil sums and $c$-differential uniformity with an application to the perturbed Gold function}, Cryptogr. Commun.  13 (2021), 891--907.

\bibitem{ZJT} Z. Tu, X. Zeng, Y. Jiang, X. Tang, {\it A class of AP$c$N power functions over finite fields of even characteristic}, arXiv (2021), {\url{https://arxiv.org/abs/2107.06464v1}}.

\bibitem{wang}  X. Wang, D. Zheng, L. Hu, {\it Several classes of PcN power functions over finite fields}, Discrete Applied Mathematics 322 (2022), 1710--182.

\bibitem{lwz} Y. Wu, N. Li, X. Zeng, {\it New P$c$N and AP$c$N functions over finite fields}, Des. Codes Cryptogr. 89 (2021), 2637--2651. 


\bibitem{XZG1}
P. Xia, S. Zhou, G. B. Giannakis, {\em Achieving the Welch Bound with Difference Sets},    IEEE Trans. Inf. Theory  51:5 (2005), 1900--1907. 
 

\bibitem{XZG2}
P. Xia, S. Zhou, G. B. Giannakis, {\em   Correction to ``Achieving the Welch bound with difference sets''}, IEEE Trans. Inf. Theory 52:7 (2006), 3359.

\bibitem{Hu}
Z. Zha, L. Hu, {\it Some classes of power functions with low $c$-differential uniformity over finite fields}, Des. Codes Cryptogr., vol. 89, 1193--1210 (2021). 

\end{thebibliography}
\end{document}